\documentclass[a4paper]{article}
 
\usepackage{graphicx}
\usepackage{amsmath,amssymb,latexsym,times,amsthm}
\usepackage{authblk}
\usepackage{url}
\usepackage{hyperref}
\usepackage[top=1in, bottom=1.25in, left=1.25in, right=1.25in]{geometry}

\usepackage{arydshln}
\usepackage{mathtools}
\usepackage{tikz}
\usetikzlibrary{arrows}
\usepackage{caption}
\usepackage{adjustbox,tabularx}
\usepackage{float}
\usepackage{verbatim}
\usepackage{comment,booktabs}
\usepackage{algorithm}
\usepackage[noend]{algpseudocode}
\usepackage{amsmath, amssymb, stmaryrd}
\usepackage{xspace}
\usepackage{xcolor}
\usepackage{needspace}
\usepackage{enumerate}
\usepackage{ifthen}
\usepackage{fancyhdr}
\usepackage{algorithmicx}
\usepackage{thmtools,thm-restate}
\usetikzlibrary{backgrounds}

\newcommand{\complClFont}[1]{\mathbf{#1}}         
\newcommand{\logicOpFont}[1]{\mathsf{#1}}         

\newcommand{\tuple}[1]{\vec{#1}}
\newcommand{\sub}{\subseteq}

\newcommand{\Fr}[1]{{\protect\ensuremath{\mathsf{Fr}(#1)}}}

\newcommand{\Po}{\mathcal{P}}
\newcommand{\MC}{\logicOpFont{MC}}
\newcommand{\MR}{\logicOpFont{MSM}}
\newcommand{\CQA}{\logicOpFont{CQA}}
\newcommand{\RC}{\logicOpFont{RC}}
\newcommand{\maxsub}{\nu}

\newcommand{\fr}{\ensuremath{\mathrm{Fr}}}

\newcommand{\LOGSPACE}{\protect\ensuremath{\complClFont{L}}\xspace}

\newcommand{\NP}{\protect\ensuremath{\complClFont{NP}}\xspace}
\newcommand{\NL}{\protect\ensuremath{\complClFont{NL}}\xspace}

\newcommand{\PTIME}{\protect\ensuremath{\complClFont{P}}\xspace}

\newcommand{\calC}{\protect\ensuremath{\mathcal{C}}}

\newcommand{\calL}{\protect\ensuremath{\mathcal{L}}}

\newcommand{\mA}{\mathfrak{A}}

\newcommand{\logicFont}[1]{\protect\ensuremath{\mathrm{#1}}
}
\newcommand{\cmdFont}[1]{\protect\ensuremath{\mathrm{#1}}\xspace}

\newcommand{\FO}{\logicFont{FO}}

\newcommand{\FOinc}{\FO(\sub)}
\newcommand{\FOkinc}[1]{\FO(#1\!\sub)}

\newcommand{\TC}{\logicFont{TC}}
\newcommand{\DTC}{\logicFont{DTC}}

\newcommand{\NLFOinc}{\FO(\sub)_{w}}
\newcommand{\game}{\mathsf{Gm}}

\newcommand{\dep}[1]{\cmdFont{dep}\!\left(#1\right)}

\newcommand{\dom}[1]{\cmdFont{dom}\!\left(#1\right)}

\newcommand{\onlyif}{\Longrightarrow}

\makeatletter
\newcommand{\xMapsto}[2][]{\ext@arrow 0599{\Mapstofill@}{#1}{#2}}
\def\Mapstofill@{\arrowfill@{\Mapstochar\Relbar}\Relbar\Rightarrow}
\makeatother

\newcommand{\al}{\alpha}
\newcommand{\safety}{\mathsf{safety}}

\renewcommand{\vec}{\overline}

\title{Complexity Thresholds in Inclusion Logic}

\author[1]{Miika Hannula}
\author[2]{Lauri Hella}

\affil[1]{University of Helsinki, Finland}
\affil[2]{Tampere University, Finland}

\theoremstyle{plain}
\newtheorem{thm}{Theorem}
\newtheorem{lem}[thm]{Lemma}
\newtheorem{cor}[thm]{Corollary}
\newtheorem{prop}[thm]{Proposition}

\theoremstyle{definition}
\newtheorem{defi}[thm]{Definition}

\begin{document}

\maketitle

\begin{abstract}
Logics with team semantics provide alternative means for logical characterization of complexity classes. Both dependence and independence logic are known to capture non-deterministic polynomial time, and the frontiers of 
tractability in these logics are relatively well understood. Inclusion logic is similar to these team-based logical formalisms with the exception that it corresponds to deterministic polynomial time in ordered models. In this article we examine connections between syntactical fragments of inclusion logic and different complexity classes in terms of two computational problems: maximal subteam membership and the model checking problem for a fixed inclusion logic formula. We show that very simple quantifier-free formulae with one or two inclusion atoms generate instances of these problems that are complete for (non-deterministic) logarithmic space and polynomial time. Furthermore, we present a fragment of inclusion logic that captures non-deterministic logarithmic space in ordered models.
   \end{abstract}

\section{Introduction}
In this article we study the computational complexity of inclusion logic. Inclusion logic was
introduced by Galliani \cite{galliani12} as a variant of dependence logic,  developed   by V\"a\"an\"anen in 2007 \cite{vaananen07}. Dependence logic is a logical formalism that extends first-order logic with novel atomic formulae $\dep{x_1,\ldots ,x_n}$ expressing that a variable $x_n$ depends on  variables $x_1, \ldots ,x_{n-1}$. One motivation behind dependence logic is to find  a unifying logical framework for analyzing dependency notions from different contexts. 
Since its introduction, versions of dependence logic have been formulated and investigated in a variety of logical environments, including propositional logic \cite{HannulaKVV18,Virtema17,YangV16}, modal logic \cite{ebbing13,vaananen08b}, probabilistic logics \cite{Durand2018}, and two-variable logics \cite{KontinenKLV14}. Recent research has also pursued  connections and applications of dependence logic to fields such as database theory \cite{HannulaK16,HannulaKV18}, Bayesian networks \cite{CoranderHKPV16}, and social choice theory \cite{Pacuit2016}. 
 A common notion underlying all these endeavours is that of team semantics. Team semantics, introduced by Hodges in \cite{hodges97}, 
  is a semantical framework where formulae are evaluated over multitudes instead of singletons of objects
as in classical logics. Depending on the application domain these multitudes  may then refer to assignment sets, probability distributions, or database tables, each having their characteristic versions of team semantics \cite{vaananen07,Durand2018,HannulaKV18}.

After the introduction of dependence logic Gr\"adel and V\"a\"an\"anen observed that team semantics can be also used to create logics for independence \cite{gradel12}. This was followed by \cite{galliani12} in which Galliani investigated 
 logical languages built upon multiple different dependency notions. Inspired by the inclusion dependencies of database theory, one of the logics introduced was inclusion logic that extends first-order logic with inclusion atoms.
 Given two sequences of variables $\tuple x$ and $\tuple y$ having same length, an inclusion atom
$\tuple x \sub \tuple y$
 expresses that the set of values of $\tuple x$ is included in the set of values of $\tuple y$. 
 Inclusion logic was shown to be equi-expressive to positive greatest-fixed point logic in \cite{gallhella13}.
  In contrast to dependence logic which is equivalent to existential second-order logic \cite{vaananen07}, and thus to non-deterministic polynomial time (\NP), this finding established inclusion logic as the first team-based based logic for polynomial time ($\PTIME$). Our focus in this article is to pursue this connection further by investigating the complexity of quantifier-free inclusion logic in terms of two computational problems:
maximal subteam membership
and model checking problems. 
In particular, we identify complexity thresholds for these problems in terms of 
 first-order definability, (non-deterministic) logarithmic space, and polynomial time.

  
The maximal subteam 
membership
problem $\MR(\phi)$ for a formula $\phi$ asks whether a given assignment is in the maximal subteam of a given team that satisfies $\phi$. This problem is closely related to the notion of a repair of an inconsistent database \cite{ArenasBC99}. A repair of a database instance $I$ w.r.t.
  some set $\Sigma$ of constraints is an instance $J$ obtained by deleting and/or adding tuples from/to $I$ such that $J$ satisfies $\Sigma$, and the difference between $I$ and $J$ is minimal according to some measure.
   If only deletion of tuples is allowed, $J$ is called a subset repair. 
It was observed in \cite{ChomickiM05} that if $\Sigma$ consists of inclusion dependencies, then for every $I$ there exists a unique subset repair $J$ of $I$; this was later generalized to arbitrary 
LAV tgds (local-as-view tuple generating dependencies)
in \cite{CateFK15}. 

The research on database repair has been mainly focused on two problems: consistent query answering and repair checking. In the former, given a query $Q$ and a database instance $I$ the problem is to compute the set of tuples that belong to $Q(J)$ for every repair $J$ of $I$.  The latter is the decision problem: is $J$ a repair of $I$ for two given database instances $I$ and $J$. The complexity of these problems for various classes of dependencies and different types of repairs has been extensively studied in the literature; see e.g. \cite{AfratiK09,ChomickiM05,KoutrisW17,CateFK15}. 
In this setting, the maximal subteam 
membership
problem can be seen as a variant of the repair checking problem: regarding a team as a (unirelational) database instance $I$ and a formula $\phi$ of inclusion logic as a constraint, an assignment is a positive instance of $\MR(\phi)$ just in case it is in the 
unique subset repair
of $I$. Note however, that in $\MR(\phi)$, the task is essentially to compute the maximal subteam from a given database instance $I$, instead of just checking that a given $J$ is  the unique subset repair of $I$.  
Note further, that using a single formula $\phi$ as a constraint is actually more general than using a (finite) set $\Sigma$ of inclusion dependencies. Indeed, as $\phi$ we can take the conjunction of all inclusions in $\Sigma$. Furthermore, using disjunctions and quantifiers, we can form constraints not expressible in the usual formalism with a set of dependencies.

The complexity of model checking in team semantics has been  studied in \cite{DurandKRV15,kontinenj13}  for dependence and independence logics. 
 For these logics increase in complexity arises particularly from disjunctions. 
  For example, model checking for a disjunction of three (two, resp.) dependence atoms is complete for $\NP$ ($\NL$, resp.), while a single dependence atom is first-order definable \cite{kontinenj13}. 
 The results of this paper, in contrast, demonstrate that the complexity of inclusion logic formulae is particularly sensitive to conjunctions. We show that $\MR(\phi)$ is complete for non-deterministic logarithmic space if $\phi$ is of the form $x\sub y$ or $x\sub y\wedge y\sub x$; for any other conjunction of (non-trivial) unary inclusion atoms $\MR(\phi)$ is complete for polynomial time. This result gives a complete characterization of the maximal subteam membership problem for conjunctions of unary inclusion atoms. 
Based on 
it
 we also prove complexity results 
for model checking of quantifier-free inclusion logic formulae.
For instance, for any
non-trivial
quantifier-free $\phi$ in which $x,y,z$ do not occur, 
model checking of $\phi \vee x\sub y$ is $\NL$-hard, while that of $\phi \vee (x\sub z \wedge y \sub z)$  is $\PTIME$-complete.

We also present a safety game for the maximal subteam membership problem. Using this game 
 we examine instances of the 
 maximal subteam membership problem in which the inclusion atoms refer to a key, that is, all inclusion atoms are of the form $x\sub y$ where $y$ is a variable which uniquely determines all the remaining variables. We give example formulae for which the thresholds between $\NL$ and $\PTIME$ drop down to $\LOGSPACE$ and $\NL$ under these assumptions. 

We conclude the paper by presenting a fragment of inclusion logic that captures $\NL$. Analogous fragments have previously been established at least for dependence logic. By relating to the Horn fragment of existential second-order logic, Ebbing et al.  define a fragment of dependence logic that corresponds to $\PTIME$ \cite{ebbing14}. The fragment presented in this paper is constructed by restricting  occurrences of inclusion atoms and universal quantifiers,  and the correspondence with $\NL$ is shown by 
using the well-known characterization of $\NL$ in terms of transitive closure logic \cite{Immerman83,Immerman88}.

\section{Preliminaries}
We generally use $x,y,z,\ldots $ for variables and $a,b,c,\ldots $ for 
elements of models.
If $\tuple p$ and $\tuple q$ are two tuples, we write $\tuple p\tuple q$ for the concatenation of $\tuple p $ and $\tuple q$. 

Throughout the paper we assume that the reader has a basic familiarity of computational complexity. 
We use the notation $\LOGSPACE$, $\NL$, $\PTIME$ and $\NP$ for the classes consisting of all problems computable in logarithmic space, non-deterministic logarithmic space,  polynomial time and non-deterministic polynomial time, respectively.

\subsection{Team Semantics}
As is customary for logics in the team semantics setting, we assume that all formulae 
are 
in negation normal form (NNF). Thus, we give the syntax of first-order logic ($\FO$) as follows:
\[
\phi::= \; t =  t' \mid \neg  t =  t'\mid R\tuple t \mid \neg R\tuple t \mid \phi \wedge \phi \mid \phi \vee \phi \mid \exists x \phi \mid \forall x \phi,
\]
where $t$ and $t'$ are terms and $R$ is a relation symbol of the underlying vocabulary.
For a first-order formula $\phi$, we denote by $\Fr{\phi}$ the set of free variables of $\phi$, defined in the usual way. 
The team semantics of $\FO$ is given in terms of the notion of a \emph{team}. Let $\mA$ be a model with domain $A$. 
An \emph{assignment} $s$ of $A$ is a function from a finite set of variables into $A$. 
We write $s(a/x)$ for the assignment that maps all variables according to $s$, except that it maps $x$ to $a$.
For an assignment 
$s=\{(x_i,a_i)\mid 1\le i\le n\}$,  
we may use a shorthand $s=(a_1, \ldots ,a_n)$ if the underlying ordering $(x_1, \ldots ,x_n)$ of the domain is clear from the context. A \emph{team} $X$ of $A$ with domain $\dom{X} = \{x_1, \ldots ,x_n\}$ is a set of assignments from $\dom{X}$ into $A$. 
For $V\sub \dom{X}$, the \emph{restriction} $X\upharpoonright V$ of a team $X$ is defined as $\{s\upharpoonright V \mid s\in X\}$. 
If $X$ is a team, $V\sub \dom{X}$, and $F: X  \rightarrow \Po(A)\setminus \{\emptyset\}$, then $X[F/x]$ denotes the team $\{s(a/x) \mid s\in X, a \in F(s)\}$. For a set $B$, $X[B/x]$ is the team $\{s(b/x) \mid s\in X, b \in B\}$. Also, if $s$ is an assignment, then by $\mA \models_s \phi$ we refer to Tarski semantics. 

\begin{defi}\label{def2}
For a model $\mA$, a team $X$ and a formula in $\FO$, the satisfaction relation $\mA \models_X \phi$ is defined as follows:
\begin{itemize}
\item $\mA \models_X \al $ if $ \forall s\in X: \mA \models_s \al $, when $\al$ is a literal,
\item $\mA \models_X \phi \wedge \psi $ if $ \mA \models_X \phi \textrm{ and } \mA \models_X \psi$,
\item $\mA \models_X \phi \vee\psi $ if $\mA \models_{Y} \phi \textrm{ and } \mA \models_{Z} \psi$ 
for some $Y,Z\subseteq X$ such that $Y\cup Z= X$,
\item $\mA \models_X  \exists  x \phi$ if $  \mA \models_{X[F/x]} \phi$ for some $F: X \rightarrow \Po(A)\setminus \{\emptyset\}$,
\item $\mA \models_X \forall x \phi $ if $  \mA \models_{X[A/x]} \phi$.
\end{itemize}
\end{defi}

If $\mA \models_X \phi$, then we say that $\mA$ and $X$ \emph{satisfy} $\phi$. 
If $\phi$ does not contain any symbols from the underlying vocabulary, in which case satisfaction of a formula does not depend on the model $\mA$, we say that $X$ 
 \emph{satisfies} $\phi$,  written $X \models \phi$, if  $\mA \models_X \phi$ for all models $\mA$ with a suitable domain (i.e., a domain that includes all the elements appearing in $X$).
If $\phi$ is a sentence, that is, a formula without any free variables, then we say that $\mA$ \emph{satisfies} $\phi$, and write $\mA \models \phi$, if $\mA \models_{\{\emptyset\}} \phi$, where $\{\emptyset\}$ is the team that consists of the empty assignment
$\emptyset$.

We say that two sentences $\phi$ and $\psi$ are equivalent, written $\phi \equiv \psi$, if $\mA \models \phi\iff \mA\models \psi$ for all models $\mA$. For two logics $\calL_1$ and $\calL_2$, we write $\calL_1\leq \calL_2$ if every $\calL_1$-sentence is equivalent to some $\calL_2$-sentence; the relations ``$\equiv$'' and ``$<$'' for $\calL_1$ and $\calL_2$ are defined in terms of ``$\leq$''  in the standard way. 

Satisfaction of a first-order formula reduces to Tarski semantics in the following way.
\begin{prop}[Flatness \cite{vaananen07}]\label{prop:flatness}
For all models $\mA$, teams $X$, and formulae $\phi\in \FO$,
\[\mA\models_X \phi \text{ iff }\mA\models_{s} \phi\text{ for all }s\in X.\]
\end{prop}
A straightforward consequence is that first-order logic is downwards closed.
\begin{cor}[Downward Closure]\label{cor:dc}
For all models $\mA$, teams $X$, and formulae $\phi\in \FO$,
\[\text{If }\mA\models_X \phi \text{ and }Y \sub X\text{, then }\mA\models_Y \phi.\]
\end{cor}

 
 \subsection{Inclusion Logic}
Inclusion logic ($\FOinc$)   is  defined as the extension of $\FO$ by inclusion atoms.

\noindent
\textbf{Inclusion atom.} Let $\tuple x$ and $\tuple y$ be two 
tuples
of variables of the same length. Then $\tuple x \sub \tuple y$ is an \emph{inclusion atom} with the satisfaction relation:
 \[\mA \models_X \tuple x \sub \tuple y\text{ if } \forall s \in X \exists s' \in X:s(\tuple x) = s'(\tuple y).\]

Inclusion logic is \emph{local}, meaning that satisfaction of a formula depends only on its free variables. Furthermore, the expressive power of inclusion logic is restricted by its \emph{union closure property} which states that satisfaction of a formula is preserved under taking arbitrary unions of teams.
\begin{prop}[Locality \cite{galliani12}]\label{prop:locality}
Let $\mA$ be a model, $X$ a team,  $\phi\in \FOinc$ a formula, and $V$ a set of variables such that $\fr(\phi)\sub V\sub \dom{X}$. Then
\[\mA\models_X \phi \iff \mA \models_{X\upharpoonright V} \phi.\] 
\end{prop}
\begin{prop}[Union Closure \cite{galliani12}]\label{prop:uc}
Let $\mA$ be a model, $\mathcal{X}$ a set of teams, and  $\phi\in \FOinc$ a formula. Then
\[\forall X\in \mathcal{X}: \mA\models_X \phi \;\onlyif\;  \mA \models_{\bigcup \mathcal{X}} \phi.\] 
\end{prop}
Note that union closure implies the \emph{empty team property}, that is, $\mA \models_{\emptyset} \phi$ for all inclusion logic formulae $\phi$. 

The starting point for our investigations is the result by Galliani and Hella \cite{gallhella13} characterizing the expressivity of  inclusion logic in terms of positive greatest fixed point logic. The latter logic is obtained from greatest fixed-point logic (the dual of least fixed point logic) by restricting to formulae in which fixed point operators occur only positively, that is, within a scope of an even number of negations. 
In finite models this positive fragment captures the full fixed point logic (with both least and greatest fixed points), and hence it follows from the famous result of Immerman \cite{immerman86} and Vardi \cite{vardi82} that inclusion logic captures 
polynomial time 
in finite ordered models.
\begin{thm}[\cite{gallhella13}]\label{thm:incgfp}
Every inclusion logic sentence is equivalent to some positive greatest fixed point logic sentence, and vice versa.
\end{thm}
\begin{thm}[\cite{gallhella13}]\label{thm:gfpcaptures}
A class $\calC$ of finite ordered models is in $\PTIME$ iff it can be defined in $\FOinc$. 
\end{thm}

\subsection{Transitive Closure Logic}
In Section \ref{sect:nl} we will  explore connections between inclusion logic and transitive closure logic, and hence we next give a short introduction to the latter. 
A $2k$-ary relation $R$ is said to be \emph{transitive} if $(\tuple a,\tuple b)\in R$ and $(\tuple b,\tuple c)\in R$  imply $(\tuple a,\tuple c)\in R$ for $k$-tuples $\tuple a,\tuple b,\tuple c$.
The \emph{transitive closure}  of a $2k$-ary relation $R$, written $\TC(R)$, is defined as the intersection of all $2k$-ary relations $S\supseteq R$ that are transitive. 
%
The transitive closure of $R$ can be alternatively defined as $R_{\infty}=\bigcup_{i=0}^{\infty} R_i$ for $R_i$  defined recursively as follows: $R_0=R$ and $R_{i+1}=R\circ R_i$ for $i>0$; here $A\circ B$ denotes the composition of two relations $A$ and $B$. 
Note that $(\tuple a,\tuple b)\in R_i$ iff there is an $R$-path of length $i+1$
from $\tuple a$ to $\tuple b$.

An assignment $s$, a model $\mA$, and a formula $\psi(\tuple x,\tuple y,\tuple z)$, where $\tuple x$ and $\tuple y$ are $k$-ary, give rise to a $2k$-ary relation defined as follows:
\[R_{\psi,\mA,s}=\{\tuple a\tuple b\in M^{2k}\mid \mA\models \psi(\tuple a,\tuple b,s(\tuple z))\}.\]

We can now define transitive closure logic. Given a term $t$, a model $\mA$, and an assignment $s$, we write $t^{\mA,s}$ for the interpretation of $t$ under $\mA,s$, defined in the usual way.

\begin{defi}[Transitive Closure Logic]
Transitive closure logic ($\TC$) is obtained by extending first-order logic with transitive closure formulae $[\TC_{\tuple x, \tuple y}\psi(\tuple x,\tuple y,\tuple z)](\tuple t_0,\tuple t_1) $
where 
$\tuple t_0$ and $\tuple t_1$ 
are $k$-tuples of terms, and  $\psi(\tuple x,\tuple y,\tuple z)$ is a formula  where $\tuple x$ and $\tuple y$ are $k$-tuples of variables. The semantics of the transitive closure formula is defined as follows:
\[\mA\models_s [\TC_{\tuple x, \tuple y}\psi(\tuple x,\tuple y,\tuple z)](\tuple t_0,\tuple t_1) \textrm{ iff } (\tuple t_0^{\mA,s} , \tuple t_1^{\mA,s})  \in \TC(R_{\psi,\mA,s}).\]
\end{defi}
Thus, $[\TC_{\tuple x, \tuple y}\psi(\tuple x,\tuple y,\tuple z)](\tuple t_0,\tuple t_1) $ is true if and only if there is a $\psi$-path from $\tuple t_0$ to $\tuple t_1$. 
It is well known that transitive closure logic captures non-deterministic logarithmic space 
in finite ordered models. In particular, this can be achieved by using only one application of the $\TC$ operator.  
We use below the notation $\min$ for the least element of the linear order, and $\tuple\min$ for the tuple $(\min,\ldots,\min)$. Similarly, $\tuple\max$ denotes the tuple $(\max,\ldots,\max)$, where $\max$ is the greatest element.
\begin{thm}[\cite{Immerman83,Immerman88}]\label{thm:TCcaptures}
A class $\calC$ of finite ordered models is in $\NL$ iff it can be defined in $\TC$. Furthermore, every $\TC$-sentence is equivalent in finite ordered models to a sentence of the form 
\[[\TC_{\tuple x, \tuple y}\alpha(\tuple x,\tuple y)](\tuple \min,\tuple \max)\] where $\alpha$ is first-order.
\end{thm}

\section{Maximal Subteam Membership}
In this section we define the maximal subteam 
membership
problem  and discuss some of its basic properties. We also define a safety game for quantifier-free inclusion logic formulae. This game will be used later to facilitate some proofs regarding the complexity of the maximal subteam membership problem. 

\subsection{Introduction}
 For a model $\mA$, a team $X$, and an inclusion logic formula $\phi$, we define $\maxsub(\mA,X,\phi)$ as the 
unique 
subteam $Y\sub X$ such that $\mA\models_Y \phi$, and $\mA\not\models_Z\phi$ if $Y \subsetneq Z \sub X$. 
Due to the union closure property 
$\maxsub(\mA,X,\phi)$
 always exists and it can be alternatively defined as the union of all subteams $Y\sub X$ such that $\mA\models_Y \phi$.  
If $\phi$ does not contain any symbols from the underlying vocabulary, then we may write $\maxsub(X,\phi)$ instead of $\maxsub(\mA,X,\phi)$. The maximal subteam
membership
problem is now given as follows.

\begin{defi} Let $\phi\in \FOinc$.
Then
$\MR(\phi)$ is the problem of determining whether $s\in \maxsub(\mA,X,\phi)$ for a given model $\mA$, a team $X$ and an assignment $s\in X$.
\end{defi}

Gr\"adel proved that for any $\FOinc$-formula $\phi$, there is a formula $\psi$ of positive greatest fixed point logic such that for any model $\mA$ and assignment $s$, $\mA\models_s\psi$ if and only if $s$ is in the maximal team of $\mA$ satisfying $\phi$ (see Theorem 24 in \cite{Gradel16}). An easy adaptation of the proof shows that $\maxsub(\mA,X,\phi)$ is also definable in positive greatest fixed point logic. Thus, it follows that every maximal subteam 
membership
problem is polynomial time computable.

\begin{lem}\label{lem:MRinP}
For every formula $\phi\in\FOinc$, $\MR(\phi)$ is in $\PTIME$.
\end{lem}

In this section
we will restrict our attention to maximal subteam problems for quantifier free formulae.
Before proceeding to our findings we need to present some auxiliary concepts and results.  The following lemmata will be useful below.
\begin{lem}\label{lem:maxsub}
Let  $\alpha,\beta\in\FOinc$, 
and let $X$ be a team of a model $\mA$. Then $\maxsub(\mA,X,\alpha\vee\beta)=\maxsub(\mA,X,\alpha)\cup \maxsub(\mA,X,\beta)$.  
\end{lem}
\begin{proof}
 For ``$\sub$'', note that by definition 
there are $Y,Z\subseteq X$ such that $Y\cup Z=\maxsub(\mA,X,\alpha\vee\beta)$,
$Y\models \alpha$ and $Z\models \beta$, 
and hence
$Y\sub \maxsub(\mA,X,\alpha)$ and $Z\sub \maxsub(\mA,X,\beta)$. For ``$\supseteq$'', note that $\maxsub(\mA,X,\alpha)\cup \maxsub(\mA,X,\beta)$ satisifes $\alpha \vee\beta$ so it must be contained by $\maxsub(\mA,X,\alpha\vee\beta)$.
\end{proof}

As an easy corollary we obtain the following lemma.
\begin{lem}\label{prop:dis}
Let $\alpha,\beta\in \FOinc$, 
 and assume that $\MR(\alpha)$ and $\MR(\beta)$ both belong to a complexity class 
$C\in\{ \LOGSPACE,\NL\}$. 
Then $\MR(\alpha\vee\beta)$ is in $C$.
\end{lem}

The maximal subteam problem for a single inclusion atom $\tuple x\sub \tuple y$ can be naturally represented using directed graphs. In this representation each assignment forms a vertex, and an assignment $s$ has an outgoing edge to another assignment $s'$ if $s(\tuple x)=s'(\tuple y)$. Over finite teams an assignment then belongs to the maximal subteam for $\tuple x\sub \tuple y$ if and only if it is connected to a cycle.\footnote{We are grateful to Phokion Kolaitis, who pointed out 
this fact to the second author in a private discussion 2016.}

\begin{lem}\label{lem:incgraph}
Let $\mA$ be a model, $X$ a finite team, $\tuple x$ and $\tuple y$ two tuples of the same length from $\dom{X}$, $s$ an assignment of $X$, and  $\alpha$ a first-order formula. Let $G=(X,E)$ be a 
directed
graph where $(s,s')\in E$ iff $s(\tuple x)=s'(\tuple y)$ and $\mA\models_{\{s,s'\}} \alpha$.
Then 
\begin{enumerate}[(a)]
\item $s\in \maxsub(\mA,X,\tuple x\sub \tuple y\wedge \alpha)   \iff G\text{ contains a path from $s$ to a cycle }$,
\item $s\in \maxsub(\mA,X,\tuple x\sub \tuple y\wedge \tuple y \sub \tuple x \wedge \alpha)   \iff G\text{ contains a path from one cycle to another via }s$\end{enumerate}
 \end{lem}

\begin{proof}
Assume for the first statement that $s\in \maxsub(\mA,X,\tuple x\sub \tuple y\wedge \alpha) $. Then there is a subteam $Y\sub X$ such that $s\in Y$ and $\mA \models _Y  \tuple x\sub \tuple y\wedge \alpha$. 
Thus for each $s'\in Y$ there exists $s''\in Y$ such that $s''=s'(\tuple x)$. Moreover, $\mA\models_{\{s',s''\}} \alpha$, whence $(s',s'')\in E$.
In particular there is a non-ending path in $G$ starting from $s$. Since $X$ is finite, this path necessarily ends in a cycle.
Conversely, assume $G$ contains a path from $s$ to a cycle. Then $\mA \models _Y  \tuple x\sub \tuple y\wedge \alpha$ where $Y$ consists of all assignments in the path and cycle. Hence, $s\in \maxsub(\mA,X,\tuple x\sub \tuple y\wedge \alpha)$.

For second statement note that, by the argument above,
$s\in \maxsub(\mA,X,\tuple y\sub \tuple x\wedge \alpha)$ if and only if $G'=(X,E^{-1})$ contains a path from $s$ to a cycle. But clearly an $E^{-1}$-path from $s$ to an $E^{-1}$-cycle is an $E$-path from an $E$-cycle to $s$. 
\end{proof}


\subsection{Safety Game}\label{sect:game}
In this section we present  a version of a safety game for the maximal subteam problem of inclusion logic. Our presentation is also related to the safety games for inclusion logic  examined in \cite{Gradel16}. We present a safety game for 
 a quadruple $(\mA,X,s,\phi)$, written $\safety(\mA,X,s,\phi)$, where $s$ is an assignment of a team $X$, and $\phi$ is a quantifier-free formula. 
 The main result of the section shows that the maximal subteam problem $\MR(\phi)$ over $X$ and $s$ can be characterized in terms of this game. 

We assume that the reader is familiar with basic terminology on trees. We associate each quantifier-free $\phi\in\FOinc$ with a labeled rooted tree $T_{\phi}$ such that the root of the tree is labeled by $\phi$ and each node labeled by $\psi_0\vee\psi_1$ or $\psi_0\wedge \psi_1$ has two children labeled by $\psi_0$ and $\psi_1$.  Notice that two different nodes may have the same label. The safety game for $(\mA,X,s,\phi)$ can now be interpreted as a pebble game where assignments of a team $X$ form a stack of pebbles of which one at a time is placed on a node of $T_{\phi}$. Legal moves of the game then consist of moving the pebble up and down through the tree, removing the pebble from a leaf, and placing a new pebble on a leaf. The starting position is to have $s$ placed on the root, and the winning condition for Player I is to arrive at a position where the game terminates. If no such position is ever reached, Player II wins.

\begin{defi}[Safety Game]
Let $\phi\in \FOinc$ be  quantifier-free, and let $s_0$ be an assignment in a team $X$ of a model $\mA$. The \emph{safety game} $\safety(\mA,X,s_0,\phi)$ has two players I and II, and the game moves consist of \emph{positions} $(s,n)$ and $(n,s)$ where $s\in X$ and $n$ is a node of $T_{\phi}$.  The game starts with the position $(s_0,r)$, where $r$ is the root, and given a position $(s,n)$, the game proceeds as follows:
\begin{enumerate}[(i)]
\item If $n$ is  labeled by a conjunction, then Player I selects a position $(s,n')$ where $n'$ is a child of $n$.
\item If $n$ is  labeled by a disjunction, then Player II selects a position $(s,n')$ where $n'$ is a child of $n$.
\item If $n$ is labeled by a literal 
$\psi$, then the game ends if 
$\mA\not\models_s \psi$.
 Otherwise, Player I selects a position $(s,n')$ such that $n$ is a descendant of $n'$.
\item If $n$ is labeled by $\tuple x \sub \tuple y$, then the game ends if there is no $s'\in X$ such that $s(\tuple x )=s'(\tuple y)$. Otherwise, Player I either
\begin{itemize}
\item selects a position $(s,n')$ such that $n$ is a descendant of $n'$, or
\item selects the position $(n,s)$.
\end{itemize}
\end{enumerate}
Given a position $(n,s)$, the game proceeds as follows:
\begin{enumerate}[(i)]\setcounter{enumi}{4}
\item Player II selects a position $(s',n)$ such that $s(\tuple x )=s'(\tuple y)$.
\end{enumerate}
Player I wins  if the game ends after a finite number of moves by the players. Otherwise, Player II wins.
\end{defi}
A strategy for a Player is a mapping $\pi$ on positions such that
\begin{itemize}
\item $\pi((s,n))\in \{(s,n')\mid n'\text{ is a child of }n\}$, for a non-leaf $n$,
\item $\pi((s,n))\in \{(s,n')\mid n\text{ is a descendant of }n'\}$, for a leaf $n$ labeled by a literal, 
\item $\pi((s,n))\in \{(s,n')\mid n\text{ is a descendant of }n'\}\cup\{(n,s)\}$, for a leaf $n$  labeled by $\tuple x \sub \tuple y$.
\item $\pi((n,s))\in \{(s',n)\mid s'\in X, s(\tuple x)=s'(\tuple y)\}$, for a leaf $n$ labeled by $\tuple x \sub \tuple y$.
\end{itemize}
Player $A\in \{\text{I,II}\}$ has a winning strategy for $\safety(\mA,X,s_0,\phi)$ if there is a strategy $\pi_A$ such that $A$ wins every game that she plays according to $\pi_A$. That is, $A$ wins any game where she selects the position $\pi_A(p)$ on her moves on $p$.
  
Note that if $\phi$ does not contain any symbols from the underlying vocabulary, the outcome of  $\safety(\mA,X,s,\phi)$ is independent of $\mA$, and thus we write $\safety(X,s,\phi)$ instead.

Next we show that the safety game above gives rise to a characterization of the maximal subteam problem. 

\begin{thm}
\label{thm:char}
Let $\phi\in \FOinc$ be quantifier-free, and let $s$ be an assignment in a team $X$ of a model $\mA$. Then $s\in \maxsub(\mA,X,\phi)$ iff Player II has a winning strategy in $\safety(\mA,X,s,\phi)$.
\end{thm}

\begin{proof}
For the ``only-if'' direction, we define top-down recursively for each node $n\in T_{\phi}$ a team $X_n$ such that
\begin{itemize}
\item $X_r:=\maxsub(\mA,X,\phi)$ for the root $r$,
\item $X_n:= \maxsub(\mA,X_{n'},\psi)$, for a child $n$ of a node $n'$ where $n$ is labeled by $\psi$.
\end{itemize}
It follows that $X_n\models \psi$ for $n$ with label $\psi$; $X_n=X_{n_0}=X_{n_1}$ for conjunction-labeled $n$ with children $n_0,n_1$; and $X_n=X_{n_0}\cup X_{n_1}$ for disjunction-labeled $n$ with children $n_0,n_1$. The strategy of Player II is now the following. If $n$ is labeled by disjunction, then $(s,n)$ is mapped to some $(s,n_i)$ where $n_i$ is a child of $n$ such that $s\in X_{n_i}$, and if $n$ is labeled by $\tuple x\sub \tuple y$, then $(n,s)$ is mapped to some $(s',n)$ such that $s(\tuple x)=s'(\tuple y)$ and $s'\in X_n$. We leave it to the reader to check that this strategy is well-defined and winning.

For the ``if'' direction, assume Player II has a winning strategy $\pi$. For a node $n$ of $T_\phi$, we let $X_n$ be the set of all  assignments $s\in X$ for which there exists a game where Player II plays according to her winning strategy and position $(s,n)$ is played at some point of the game. Consider any assignment $s$ from $X_n$ for a node $n$ labeled by $\tuple x\sub \tuple y$. This means there is a game where position $(s,n)$ is played, and thus also a game where $(n,s)$, and furthermore $\pi((n,s))=(n,s')$ is played. Consequently, an assignment $s'\in X_n$ exists such that $s(\tuple x)=s'(\tuple y)$. 
By analogous reasoning we obtain that $X_n\models \psi$ for all other types of nodes $n$ with label $\psi$, too. 
Furthermore, $X_n=X_{n_0}=X_{n_1}$ for conjunction-labeled $n$ with children $n_0,n_1$, and $X_n=X_{n_0}\cup X_{n_1}$ for disjunction-labeled $n$ with children $n_0,n_1$. In particular, $X_r\models \phi$ and $s\in X_r$, and hence $s\in \maxsub(X,\phi)$.
\end{proof}

Given that $X$ is finite, it makes sense to consider bounded length restrictions of the safety game. 
We let $\safety_k(\mA,X,s,\phi)$ denote the version of $\safety(\mA,X,s,\phi)$ in which, starting position $(s,r)$ excluded, positions of the form $(s,n)$, i.e., pairs whose left element is an assignment and right element a node, are played at most $k$ times. Player I wins $\safety_k(\mA,X,s,\phi)$ if the game terminates before such assignment-node pairs have been played $k$ times. Otherwise, if exactly $k$ plays of such nodes appear, Player II wins. The next lemma will be useful 
later.

\begin{lem}\label{apulemma}
Let $\phi\in \FOinc$ be quantifier-free and such that $T_\phi$ has $k$ nodes, and let $s$ be an assignment of a team $X$ that is of size $l$. Then Player II has a winning strategy for $\safety(\mA,X,s,\phi)$ iff she has a winning strategy for $\safety_{k\cdot l}(\mA,X,s,\phi)$.
\end{lem}
\begin{proof}
By the end of $\safety_{k\cdot l}(\mA,X,s,\phi)$, positions of the form $(s,n)$, the root position included, have occurred $k\cdot l+1$ many times, i.e., some position $(s,n)$ has occurred twice. Every time such a repetition is encountered, we may assume that we continue the game from the first occurrence of $(s,n)$. Since the strategy of Player II is safe for $k\cdot l$ assignment-node moves, we conclude that $\safety(\mA,X,s,\phi)$ never terminates. Hence, Player II wins.
\end{proof}


\section{Complexity of Maximal Subteam Membership}
Next we examine the computational complexity of maximal subteam membership. First in Section \ref{sect:arb} we investigate this problem over arbitrary teams, and then in Section \ref{sect:arb}
 we restrict attention to inputs in which the inclusion atoms refer to a key. In Section \ref{sect:CQA} we discuss the implications of our results to consistent query answering.
\subsection{Arbitrary Teams}\label{sect:arb}
 We give a complete characterization of the maximal subteam problem for arbitrary conjunctions of unary inclusion atoms. A \emph{unary} inclusion atom is an atom of the form $x \sub y$ where $x$ and $y$ are single variables. The characterization is given in terms of inclusion graphs.

\begin{defi}
Let $\Sigma$ be a set of unary inclusion atoms over variables in $V$. Then the \emph{inclusion graph} of $\Sigma$ is defined as $G_{\Sigma}=(V,E)$ such that $(x,y)\in E$ iff $x \neq y$ and $x \sub y$ appears in $\Sigma$.
\end{defi}

We will now prove the following theorem.

\begin{thm}\label{thm:trikotomia}
Let $\Sigma$ be a finite set of unary inclusion atoms, and let $\phi$ be the conjunction of all atoms in $\Sigma$. Then $\MR(\phi)$ is
\begin{enumerate}[(a)]
\item trivially true if $G_{\Sigma}$ has no edges,
\item $\NL$-complete if $G_{\Sigma}$ has an edge $(x,y)$ and no other edges except possibly for its inverse $(y,x)$,
\item\label{kolmas}  $\PTIME$-complete otherwise. 
\end{enumerate}
\end{thm}

The first statement above  follows from the observation that $\MR(\phi)$ is true for all inputs if $\phi$ is a conjunction of trivial inclusion atoms $x\sub x$. The second statement is shown by
 relating to graph reachability. Given a directed graph $G=(V,E)$ and two vertices $a$ and $b$, the problem REACH is to determine whether $G$ contains a path from $a$ to $b$. 
 This problem is a well-known complete problem for $\NL$, and it will also be applied later in Section \ref{sect:key} where the complexity of $\MR(x\sub y \wedge u\sub v)$ over teams with keys $y$ and $v$ is examined. 
\begin{lem}\label{thm:NLarb}
$\MR(x\sub y)$ and $\MR(x \sub y \wedge y \sub x)$
 are $\NL$-complete.
\end{lem}
\begin{proof}
 \textbf{Hardness.} We give a logarithmic space many-one reduction from  
REACH. 
Let $G=(V,E)$ be a graph, and let $a,b\in V$.
W.l.o.g. we can assume 
 $G$ has no cycles. Define $E'$ as the extension of $E$ with an extra edge $(b,a)$. Then $b$ is reachable from $a$ in $G$ if and only if $a$ belongs to a cycle in $G'=(V,E')$. 
We reduce from $(G,a,b)$ to a team $X=\{s_{d,c}\mid (c,d)\in E'\}$ 
  where $s_{u,v}$  maps $(x,y)$ to $(u,v)$ (see Fig. \ref{kuvaB}). 
  By Lemma \ref{lem:incgraph}, 
$b$ is reachable from $a$ if and only if 
$s_{a,b}\in \maxsub(X,\phi)$, where 
   $\phi$ is either $x\sub y$ or $x\sub y\wedge y\sub x$.

\noindent \textbf{Membership.} 
By Lemma \ref{lem:incgraph} $\MR(x\sub y)$ and $\MR(x \sub y \wedge y \sub x)$ reduce to reachability variants that are clearly in $\NL$.
\end{proof}

\begin{figure}
\centering
\begin{tabular}{ccc}
\adjustbox{valign=m}{
\begin{tikzpicture}[->,>=stealth,shorten >=1pt,auto,node distance=1.5cm,
 main node/.style={circle,draw,minimum size=0.6cm,draw,font=\sffamily\bfseries}]

  \node[main node] (1)  {$a$};
 \node[main node] (2) [below of=1] {$0$};
   \node[main node] (3) [below of=2] {$2$};
     \node[main node] (5) [left  of=3] {$1$};
   \node[main node] (4) [right of=3] {$b$};

\path[every node/.style={sloped,anchor=south,auto=false}]

    (1) edge (2)
    (2) edge (4)
    (2) edge (5)
    (5) edge (3)
    (4) edge [thick, dotted, bend right] (1)
    (2) edge (3);
      
\end{tikzpicture}}
&
\qquad
$\mapsto$
\vspace{0cm}
&
\qquad
\begin{tabular}{ccc}
\cmidrule[.5pt]{2-3}
   &$x$ & $y$ \\\cmidrule{2-3}
      $\bullet$ &$a$ & $b$\\
   $\circ$ & $0$ & $a$ \\
   & $1$ & $0$ \\
    &$2$ & $0$ \\
    $\circ$&$b$ & $0$ \\
    &$2$ & $1$  \\\cmidrule[.5pt]{2-3}
    &&\multicolumn{1}{c}{\makebox[0pt]{$\phi= \begin{cases}x\sub y\\x\sub y\wedge y\sub x\end{cases}$}}
    \end{tabular}
    \end{tabular}
\caption{Reduction from REACH to $\MR(\phi)$. The black circle marks the input assignment and all the circles together mark a subteam satisfying $\phi$. \label{kuvaB}}
\end{figure}
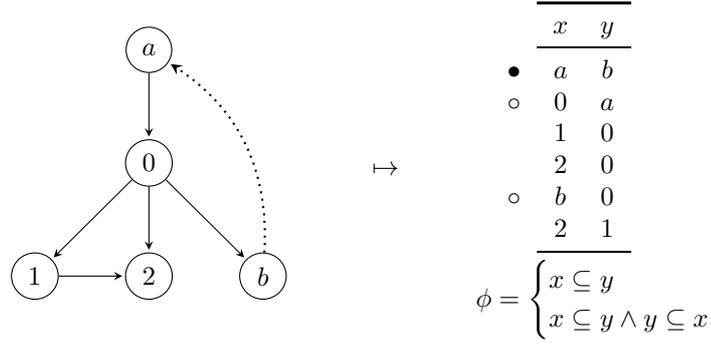



Next we turn to the third statement of Theorem \ref{thm:trikotomia}. Recall that membership in $\PTIME$ follows directly from Lemma \ref{lem:MRinP}.
For $\PTIME$-hardness we reduce from the 
monotone
circuit value problem (see, e.g., \cite{Vollmerbook}).  The proof essentially follows from the following lemma. 
\begin{lem}\label{lem:eka}
\label{thm:P}
$\MR(x\sub z\wedge y\sub z)$, $\MR(x\sub y\wedge y\sub z)$, and $\MR(x\sub y\wedge x\sub z)$
 are $\PTIME$-complete.
\end{lem}

\begin{proof}
Let $\phi$ be either $x\sub z\wedge y\sub z$, $x\sub y\wedge y\sub z$, or $x\sub y\wedge x\sub z$. 
We give a logarithmic-space many-one reduction to $\MR(\phi)$ from the monotone circuit value problem (MCVP). Given a Boolean word $w\in \{\top,\bot\}^n$, and a Boolean circuit $C$ with $n$ inputs, one output, and gates with labels from $\{\text{AND,OR}\}$, this problem is to determine whether $C$ outputs $\top$. If $C$ outputs $\top$ on $w$, we say that it \emph{accepts} $w$. W.l.o.g. we may assume that the in-degree of each AND and OR gate is $2$. 
 We annotate each input node by its corresponding input $\top$ or $\bot$, 
 and each gate by some distinct number $i\in \mathbb{N}$.
 Then each gate has two child nodes $i_L,i_R$ that are either natural numbers or input values from $\{\top,\bot\}$. Next we construct a team $X$ whose values consist of node annotations $i,\top,\bot$ and distinct copies $c_i$ of AND gates $i$. The team $X$ is constructed by the following rules (see Fig. \ref{kuvaA}): 
\begin{itemize}
\item add $s_0\colon (x,y,z) \mapsto (1,\top,\top)$ where $1$ is the output gate,
\item for AND gates $i$,  add $s_{i,0}\colon (x,y,z)\mapsto (i_L,i,c_i)$, $s_{i,1}\colon  (x,y,z)\mapsto (i_R,c_i,i)$, and $s_i:(x,y,z)\mapsto (c_i,\top,\top)$,
\item for OR gates $i$, add $s_{i,L}\colon  (x,y,z)\mapsto (i_L,i,i)$ and $s_{i,R}\colon  (x,y,z) \mapsto (i_R,i,i)$.
\end{itemize}
We will show that $C$ accepts $w$ iff $s_0\in \maxsub(X,\phi)$. For the only-if direction we actually show a slightly stronger claim. That is, we show that $s_0\in \maxsub(X,\phi)$ is implied even if $\phi$ is the conjunction of all unary inclusion atoms between $x,y,z$.
 
Assume first that $C$ accepts $w$. We show how to build a subteam $Y\sub X$ such that it includes $s_{0}$ and satisfies all unary inclusion atoms between $x,y,z$.  
 First construct a subcircuit $C'$ of $C$ recursively as follows: add the output gate $\top$ to $C'$; for each added AND gate $i$, add both child nodes of $i$; for each added OR gate $i$, add a child node of $i$ that is evaluated true under $w$. In other words, $C'$ is a set of paths from the output gate to the input gates that witnesses the assumption that $C$ accepts $w$. The team $Y$ will now list the auxiliary values $c_i$ and the gates of $C'$ in each column $x,y,z$. We construct $Y$ by the following rules: 
 \begin{itemize}
 \item add $s_0$,
\item for AND gates $i$ in $C'$, add $s_{i,0}$, $s_{i,1}$, and $s_i$,
\item for OR gates $i$ in $C'$, add $s_{i,P}$ iff $i_P$ is in $C'$, for $P=L,R$.
\end{itemize}
 It follows from the construction that $Y$ satisfies all unary inclusion atoms between $x,y,z$.

Vice versa, consider the standard semantic game between Player I and Player II on the given 
circuit $C$ and input word $w$.
This game starts from the output gate $1$, and at each AND (OR, resp.) gate $i$ Player I (Player II, resp.) selects the next node from its two child nodes $i_L$ and $i_R$. Player II wins iff the game ends at an input node that is true. By the assumption that $s_0\in \maxsub(X,\phi)$ we find a team $Y$ that contains $s_0$ and satisfies $\phi$. Note that $Y$ cannot contain any assignment that maps $x$ to $\bot$. For showing that $C$ accepts
$w$
it thus suffices to show that Player II has a strategy which imposes the following restriction: for each visited node annotated by $i$, 
  we have $s(x)=i$ for some $s\in Y$. 
  At start this holds by the assumption that $s_0 \in Y$. 
 Assume that $i$ is any gate with $s\in Y$ such that $s(x)=i$. 
If $\phi$ is $x\sub z\wedge y\sub z$, we have two cases. If $i$ is an OR gate then we find $s'$ from $Y$ with $s'(y)=s'(z)=i$. Then the strategy of Player II is to select the gate $s'(x)$ as her next step. If $i$ is an AND gate, an application of $x\sub z$ gives $s'$ from $Y$ with $s'(z)=i$.  Then $s'(y)=c_i$, which means that further application of $y\sub z$ yields $s''$ from $Y$ with $s''(z)=c_i$ and hence $s''(y)=i$. Now $\{s'(x),s''(x)\}=\{i_L,i_R\}$, and thus the claim holds for either selection by Player I. The induction step is analogous for the cases where $\phi$ is $x\sub y\wedge y\sub z$ or $x\sub y\wedge x\sub z$.
This concludes the proof.
 \end{proof}

\begin{figure}[h!]
\centering
\begin{tabular}{ccc}
\adjustbox{valign=m}{
\begin{tikzpicture}[->,>=stealth,auto,node distance=1.1cm,
 every node/.style={circle,draw,minimum size=0.6cm,draw,font=\sffamily\bfseries}]

  \node[label={$1$}]  {$\wedge$}
  child {node[label={$2$}]  {$\vee$}
  	child{node[label={}] {$\top$}}
  	 child{node[label={}] {$\wedge$}}
  	}
  child {node[label={$3$}]  {$\vee$}
  	child {node[label={$4$}] {$\wedge$}
  		child {node[label={$$}]  {$\bot$}}
  		child {node[label={$6$}]  (5) {$\vee$}}
  		}
  	child	{node[label={$5$}] (3){$\wedge$}
  		child {node (6) {$\vee$}
  		  	child {node[label={$$}] {$\bot$}}
  			child {node[label={$$}] {$\top$}}
  			}
  		child {node[label={$$}] (7) {$\top$}
  			}
  		}
	}
  		;
  		
  \path[every node/.style={sloped,anchor=south,auto=false}]
  (3) edge (5);
\end{tikzpicture}}
&

$\mapsto$
\qquad
&
\begin{tabular}{cccc}
\cmidrule[.5pt]{2-4}
&$x$&$y$&$z$\\\cmidrule{2-4}
$\bullet$&$1$&$\top$&$\top$\\\cdashline{2-4}
$\circ$&$2$&$1$&$c_1$\\
$\circ$&$3$&$c_1$&$1$\\\cdashline{2-4}
$\circ$&$\top$&$2$&$2$\\
$$&$2$&$2$&$2$\\\cdashline{2-4}
$$&$4$&$3$&$3$\\
$\circ$&$5$&$3$&$3$\\\cdashline{2-4}
$$&$\bot$&$4$&$c_4$\\
$$&$6$&$c_4$&$4$\\\cdashline{2-4}
$\circ$&$6$&$5$&$c_5$\\
$\circ$&$\top$&$c_5$&$5$\\\cdashline{2-4}
$$&$\bot$&$6$&$6$\\
$\circ$&$\top$&$6$&$6$\\\cdashline{2-4}
$\circ$&$c_1$&$\top$&$\top$\\
$$&$c_4$&$\top$&$\top$\\
$\circ$&$c_5$&$\top$&$\top$\\
\cmidrule[.5pt]{2-4}
\multicolumn{4}{c}{\makebox[0pt]{$\phi= \begin{cases}x\sub z\wedge y\sub z\\x\sub y\wedge y\sub z\\ x\sub y\wedge x\sub z\end{cases}$}}
\end{tabular}

    \end{tabular}
\caption{MCVP and $\MR(\phi)$ \label{kuvaA}}
\end{figure}
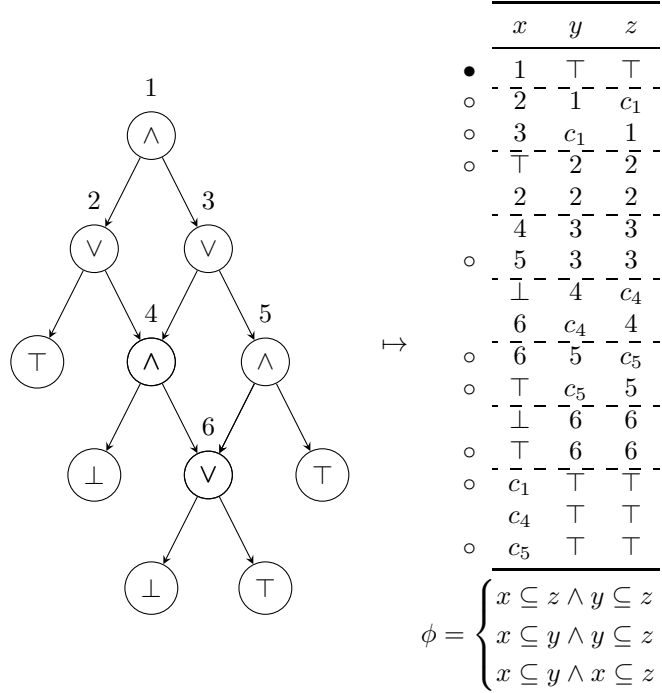

The third statement of Theorem \ref{thm:trikotomia} now follows. Any $G_\Sigma$ not covered by the previous lemma has a subgraph of a form depicted in Fig. \ref{fig:graphs}. Of these 
$G_1$--$G_3$ 
were considered above, and the reduction for $G_4$ is essentially identical to that for $G_1$; take a new variable for the new target node and insert values identical to those of $z$. 
 Additionally, for each node in $G_{\Sigma}$ but not in $G_i$ take a copy of any column in the team. That this suffices follows from the arguments of the previous lemmata; in particular, from the fact that any true MCVP instance generates a subteam that satisfies all possible unary inclusion atoms between variables.

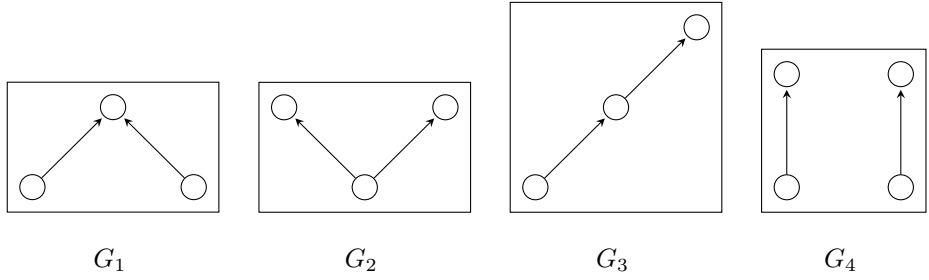
\begin{figure}
\centering
\begin{tabular}{cccc}
\adjustbox{valign=b}{
\begin{tikzpicture}[framed,->,>=stealth,shorten >=1pt,auto,node distance=1.5cm,
 main node/.style={circle,draw,minimum size=0.1cm,draw,font=\sffamily\bfseries}]

  \node[main node] (1)  {$$};
 \node[main node] (2) [below left of=1] {$$};
   \node[main node] (3) [below right of=1] {$$};

\path[every node/.style={sloped,anchor=south,auto=false}]

    (2) edge (1)
    (3) edge (1);
      
\end{tikzpicture}}
&
\adjustbox{valign=b}{
\begin{tikzpicture}[framed,->,>=stealth,shorten >=1pt,auto,node distance=1.5cm,
 main node/.style={circle,draw,minimum size=0.1cm,draw,font=\sffamily\bfseries}]

  \node[main node] (1)  {$$};
 \node[main node] (2) [above left of=1] {$$};
   \node[main node] (3) [above right of=1] {$$};

\path[every node/.style={sloped,anchor=south,auto=false}]

    (1) edge (2)
    (1) edge (3);
      
\end{tikzpicture}}
&
\adjustbox{valign=b}{
\begin{tikzpicture}[framed,->,>=stealth,shorten >=1pt,auto,node distance=1.5cm,
 main node/.style={circle,draw,minimum size=0.1cm,draw,font=\sffamily\bfseries}]

  \node[main node] (1)  {$$};
 \node[main node] (2) [ above right of=1] {$$};
   \node[main node] (3) [ above right of=2] {$$};

\path[every node/.style={sloped,anchor=south,auto=false}]

    (1) edge (2)
    (2) edge (3);
      
\end{tikzpicture}}
&
\adjustbox{valign=b}{
\begin{tikzpicture}[framed,->,>=stealth,shorten >=1pt,auto,node distance=1.5cm,
 main node/.style={circle,draw,minimum size=0.1cm,draw,font=\sffamily\bfseries}]

  \node[main node] (1)  {$$};
 \node[main node] (2) [above  of=1] {$$};
   \node[main node] (3) [ right of=1] {$$};
   \node[main node] (4) [ right of=2] {$$};

\path[every node/.style={sloped,anchor=south,auto=false}]

    (1) edge (2)
    (3) edge (4);
      
\end{tikzpicture}}\\
&&&\\
$G_1$&$G_2$&$G_3$&$G_4$

\end{tabular}
\caption{Subgraphs of $G_{\Sigma}$ \label{fig:graphs}}
\end{figure}



Considering disjunctions, observe that $\MR$ over a disjunction of unary inclusion atoms is either trivially true or $\NL$-complete. 
For membership in $\NL$, see Lemma \ref{prop:dis}. 
For $\NL$-hardness one may use the reduction of Lemma \ref{thm:NLarb} for 
 $\MR(x \sub y \vee y\sub x)$. Provided that another non-trivial inclusion atom $u \sub v$ appears in the disjunction, then $\{u,v\}\not\sub\{x,y\}$ and the values for $u,v$ can be defined in such a way  
 that the maximal subteam for $u \sub v$ is empty.

\begin{cor}\label{cor:NL}
Let $\Sigma$ be a finite set of unary inclusion atoms, and let $\phi$ be the disjunction of all atoms in $\Sigma$. Then $\MR(\phi)$ is
\begin{enumerate}[(a)]
\item trivially true if $\Sigma$ contains a trivial inclusion atom,
\item $\NL$-complete otherwise. 
\end{enumerate}
\end{cor}

Note that the results of this section generalize to inclusion atoms of higher arity, obtained by replacing variables $x$ with tuples $\tuple x$ such that all pairs of distinct tuples have no common variables. More complex cases arise if the tuples are allowed to overlap. 

\subsection{Teams with a Key}\label{sect:key}
In relational database management  inclusion atoms usually appear in form of a foreign key that has the purpose of securing referential integrity upon databases. 
  In this section we investigate the maximal subteam
problems in a framework where inclusion atoms correspond to uni-relational foreign keys. That is, we consider inclusion atoms $x\sub y$ over teams on which the variable $y$ is a key. Given this additional restriction, we observe that the maximal subteam and model checking problems both collapse to lower computational levels.  The complexity of the maximal subteam problem for $x\sub y$ collapses from $\NL$ to $\LOGSPACE$, and for $x\sub z \wedge y \sub z$ it collapses from $\PTIME$ to $\NL$.
  
   \begin{defi}
Let $X$ be a team over a set of variables $V$, and let $U$ be a subset of $V$. Then $U$ is a \emph{key} on $X$ if for all $s,s'\in X:s(U)=s'(U)\implies s=s'$.
\end{defi}

First we  show that maximal subteam for single inclusion atom is  $\LOGSPACE$-complete if the inclusion atom refers to a key. 

\begin{thm}
\label{thm:L}
$\MR(x\sub y)$ over teams 
for which $y$ is a key is
$\LOGSPACE$-complete.
\end{thm}

\begin{proof}
\textbf{Hardness.} We show a first-order many-one reduction from a deterministic variant of the reachability problem. Given a directed graph $G$ and its two vertices $a$ and $b$, this problem is to determine whether there is a deterministic path from $a$ to $b$. A deterministic path is such that for every edge $(i,j)$ in the path there is no other edge in $G$ going out of $i$. Similarly to above, and with respect to first-order reductions, w.l.o.g. we may assume that the out-degrees of $a$ and $b$ are $1$ and $G$ has no cycles except for a self-loop on $b$. An instance of  deterministic rechability is reduced to a team $X$ over $\{x,y\}$ by applying the formation rule:
\begin{itemize}
\item For each edge $(u,v)$ that is the only edge going out from $u$, add the assignment $(y,x)\mapsto (u,v)$ to $X$.
\end{itemize}
 For instance, the instance illustrated in Fig. \ref{kuvaB} reduces to the team that consists of the three 
assignments 
$(y,x)\mapsto \{(a,0),(1,2),(b,b)\}$. We show that the input instance admits positive answer iff $m\in \maxsub(X,x\sub y)$, where $m$ is the 
assignment
that corresponds to the only edge going out from $s$.
 

By Lemma \ref{lem:incgraph}, $s\in \maxsub(X,x\sub y)$ iff $s$ is connected to a cycle in the graph $G_X=(X,E)$, where $(s_0,s_1)\in E$ iff $s_0(y)=s_1(x)$. Since the only cycle in $G$ is the self-loop on $b$, accordingly the only cycle in $G_X$
is the self-loop on the mapping that corresponds to the self-loop on $b$. Furthermore, the team formation rule for $X$ implies that any path in $G_X$ corresponds to a deterministic path in $G$, and vice versa. Consequently, $s\in \maxsub(X,x\sub y)$ iff there is a deterministic path from $a$ to $b$.

Lastly, we note that the reduction to $X$ and $s$ is clearly first-order, and that $y$ is a key on $X$. This concludes the hardness proof. 

\noindent
\textbf{Membership.}  Since $y$ is a key on $X$, all paths in the graph $G$ given in Lemma \ref{lem:incgraph} are deterministic.  The existence of a path from a given assignment of $X$ to a cycle in $G$ is then confirmed if the unique path outgoing this assignment is of length at least $|X|$. This process can be clearly executed in $\LOGSPACE$. 

\end{proof}


Next we consider maximal subteam membership for a conjunction of two inclusion atoms which both refer to a key on the input team. Again, we utilize 
 reachability for showing $\NL$-hardness. In contrast to our earlier reduction from the same problem, 
we now have both an extra restriction and an extra allowance: teams must be of the sort where inclusion atoms point to some key, however instead of one inclusion atom now two inclusion atoms are available. 

\begin{thm}
\label{thm:NLL}
$\MR(x\sub z\wedge y\sub z)$ over teams 
for which $z$ is a key
is $\NL$-complete.
\end{thm}
\begin{proof}
\textbf{Hardness.} As above, it suffices to 
give
a logarithmic-space many-one reduction from reachability. 
Recall that reachability is the problem to decide, given a directed graph $G$ with two nodes $a$ and $b$,  whether there is a path from $a$ to $b$. 
 W.l.o.g. we  restrict attention to instances in which all vertices have out-degree at most two. Since $\NL$ is closed under complement, we may reduce  to the complement of $\MR(x\sub z\wedge y\sub z)$ over teams with a key $z$. 
 W.l.o.g. we may assume that $b$ has out-degree $0$ and that all the other nodes have out-degree at least $1$. The team $X$ over $\{x,y,z\}$ is now constructed by the following formation rule (see also Fig. \ref{kuvaC}). For all vertices $i\neq b$:
\begin{itemize}
\item if $i$ has single outgoing edge $(i,j)$, add the assignment $s_i:(x,y,z)\mapsto (j,j,i)$ to $X$;
\item if $i$ has two outgoing edges $(i,j)$ and $(i,k)$, add the assignment $s_i:(x,y,z)\mapsto (j,k,i)$ to $X$.
\end{itemize}
Observe that $z$ is a key on $X$.
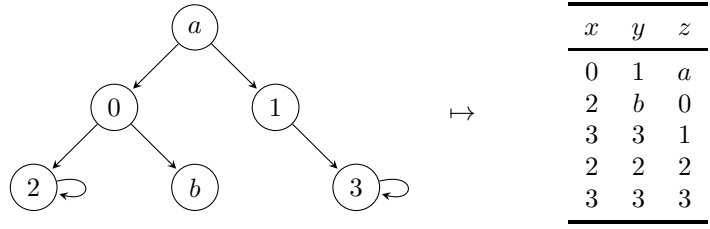
\begin{figure}
\centering
\begin{tabular}{ccc}
\adjustbox{valign=m}{
\begin{tikzpicture}[->,>=stealth,shorten >=1pt,auto,node distance=1.5cm,
 main node/.style={circle,draw,minimum size=0.6cm,draw,font=\sffamily\bfseries}]

  \node[main node] (1)  {$a$};
 \node[main node] (2) [below left of =1] {$0$};
  \node[main node] (5) [below right of =1] {$1$};
    \node[main node] (6) [below right of =5] {$3$};
   \node[main node] (3) [below left of=2] {$2$};
   \node[main node] (4) [below right of=2] {$b$};
\path[every node/.style={sloped,anchor=south,auto=false}]

    (1) edge (2)
    (2) edge (4)
   (1) edge (5)
    (5) edge (6)
    (3) edge [loop right]  (3)
        (6) edge [loop right]  (6)
    (2) edge (3);
      
\end{tikzpicture}}
&

$\mapsto$
\vspace{0cm}
&
\qquad
    \begin{tabular}{ccc}
    \toprule
    $x$ & $y$ & $z$ \\
    \midrule
    $0$ & $1$ & $a$ \\
    $2$ & $b$ & $0$ \\
    $3$ & $3$ & $1$ \\
    $2$ & $2$ & $2$ \\
    $3$ & $3$ & $3$ \\
    \bottomrule
    \end{tabular}
    \end{tabular}
\caption{Reachability and $\MR(x\sub z\wedge y\sub z)$ over teams with a key $z$\label{kuvaC}}
\end{figure}
By Theorem \ref{thm:char} it suffices to show that Player II has no winning strategy in $\safety(X,s_a,x\sub z\wedge y\sub z)$ iff $a$ is connected to $b$. 

First note that since $z$ is a key, Player II has a pre-determined strategy: either there is no assignment to choose or else Player II has to always 
pick
the only possible assignment which keeps the game running. For instance, given an assignment-node position $(s,n)$ where $n$ is labeled by $x\sub z$, Player II can only counter by selecting the position $(n,s')$ where $s'$ is the only assignment of $X'$ such that $s'(y)=s(x)$, provided that such an assignment exists. Hence, $a$ is connected to $b$ iff Player I wins some instance of $\safety(X,s_a,x\sub z\wedge y\sub z)$. Furthermore, Player I wins iff the game reaches a position $(s,n)$ where $s$ maps either $x$ or $u$  to $b$. A sequence of positions that is winning for Player I now generates a path from $a$ to $b$, formed by following the edges which correspond to Player I's selections between $x\sub z$ and $y\sub z$. Conversely, a winning sequence for Player I can be found by moving to  $x\sub z$ ($y\sub z$, resp.) whenever the values of $(z,x)$ ($(z,y)$, resp.) form the next edge on the path.

\noindent
\textbf{Membership.} Since $\NL$ is closed under complement, it suffices to describe a $\NL$ procedure for deciding the complement problem for $\MR(x\sub z\wedge y\sub z)$ over teams with a key $z$. 
By Theorem \ref{thm:char} and Lemma \ref{apulemma}, $s\notin \maxsub(X,x\sub z\wedge y\sub z)$ iff Player II has no winning strategy in $\safety_{3\cdot l}(X,s,x\sub z\wedge y\sub z)$, where $l=|X|$. Recall that $z$ is a key on the given team $X$, and hence the strategy of Player II is pre-determined. 
Hence, it suffices to guess the choices of Player I and accept iff the game terminates before $3\cdot l$ assignment-element pairs have been played in the game.
\end{proof}

Adapting the reasoning behind Corollary \ref{cor:NL} we obtain a corollary for disjunction, too.
\begin{cor}
$\MR(x\sub z \vee y \sub z)$ over teams 
for which $z$ is a key,
is $\LOGSPACE$-complete.
\end{cor}

\subsection{Consistent Query Answering}\label{sect:CQA} The maximal subteam problem has a close connection to database repairing which provides a framework for managing inconsistency in databases.  An inconsistent database is a database that does not satisfy all the integrity constraints that it is supposed to satisfy. Inconsistency may arise, e.g., from data integration where the task is to bring together data from different sources. Often in practise inconsistency is handled through data cleaning which is the process of identifying and correcting inaccurate data records from databases. An inherent limitation of this approach is its inability to avoid arbitrary choices as consistency can usually be restored in a number of ways. The approach of database repair is to tolerate inconsistencies in databases and investigate reliable answers to queries. 

A \emph{database} is an interpretation of a relational vocabulary $\sigma=\{R_1,\ldots ,R_n\}$ in which each $R_i$ is associated with an arity $\#R_i$. Given a (finite) set $\Sigma$ of integrity constraints, a database $D$ is called \emph{inconsistent} (w.r.t. to $\Sigma$) if $D\not\models \Sigma$, and \emph{consistent} otherwise. Given a partial order $\leq$ on databases over a fixed $\sigma$, and a set $\Sigma$ of integrity constraints, a \emph{repair} of an inconsistent database $I$ is a database $D$ such that it is consistent and all $D'< D$ are inconsistent. The database $D$ is called a $\oplus$-\emph{repair} if the partial order is defined in terms of symmetric difference: $D\leq D'$ if  $D\oplus I \sub D' \oplus I$. If additionally $D$ is a subset (superset, resp.) of $I$, then $D$ is called a \emph{subset-repair} (\emph{superset-repair}, resp.). An \emph{answer} to a first-order query $q=\psi(x_1, \ldots ,x_n)$ on a database $D$ is any $(a_1, \ldots ,a_n)$ such that  $D$ satisfies $\psi(a_1, \ldots ,a_n)$, and a \emph{consistent answer} on an inconsistent database $I$ is any value $(a_1, \ldots ,a_n)$ such that each repair $D$ of $I$ satisfies $\psi(a_1, \ldots ,a_n)$. Let $*\in \{\oplus$, subset, superset$\}$ and let $\Sigma$ be a set of integrity constraints. 
The \emph{$*$-repair checking problem w.r.t. $\Sigma$} ($*$-$\RC(\Sigma)$) is to determine, given two databases $D$ and $I$, whether $D$ is a $*$-repair of $I$. Let also $q$ be a Boolean query. The \emph{$*$-consistent query answering problem w.r.t. $\Sigma$ and $q$} ($*$-$\CQA(\Sigma,q)$) is to determine, given an inconsistent database $I$, whether $q$ is true in every $*$-repair of $I$. Inclusion dependencies are a special case of so-called LAV tgds that are first-order formulae of the form 
$\forall \tuple x (\psi(\tuple x \to \exists \tuple y \theta(\tuple x,\tuple y))$
 where $\psi$ is a single relational atom and $\theta$ is a conjunction of relational atoms. For LAV tgds these problems are solvable in polynomial time.
 \begin{thm}[\cite{CateFK15}]
 Let $*\in \{\oplus$, subset, superset$\}$, let $\Sigma$ be a set of LAV tgds, and let $q$ be a conjunctive query. The $*$-repair checking problem w.r.t. $\Sigma$ and the $*$-consistent query answering problem w.r.t. $\Sigma$ and $q$ are both solvable in polynomial time.
 \end{thm}
 Furthermore, it is known that weakly acyclic collections of LAV tgds enjoy subset-repair checking in logarithmic space \cite{AfratiK09}. Nevertheless, it seems not much attention in general has been devoted to complexity thresholds within polynomial time. Our results can thus be seen as steps toward this direction as the trichotomy in Theorem \ref{thm:trikotomia}  extends to repair checking and consistent query answering. 
 Let $IC$ be a collection of finite sets of integrity constraints and let $C$ be a complexity class. 
 We say that the consistent query answering problem is $C$-complete for $IC$  if for all  $\Sigma\in IC$,  $*$-$\CQA(\Sigma,q)$ is in $C$ for all Boolean conjunctive queries $q$ and $C$-complete for some such $q$. 
 
\begin{cor}
 Let $*\in \{\oplus$, subset$\}$. 
  The subset-repair checking problem 
   and the $*$-consistent query answering problem for finite sets $\Sigma$ of unary unirelational inclusion dependencies are
 \begin{enumerate}[(a)]
\item first-order definable if $G_{\Sigma}$ has no edges,
\item $\NL$-complete if $G_{\Sigma}$ has an edge $(x,y)$ and no other edges except possibly for its inverse $(y,x)$,
\item $\PTIME$-complete otherwise. 
\end{enumerate}
\end{cor}
Since $\NL$ and $\PTIME$ are closed under complement, we may consider  the complement of subset-repair checking. For the upper bounds note that $D$ is a repair of $I$ if and only if $D$ satisfies $\Sigma$ (a first-order property) and no tuple in $I\setminus D$ is in the unique subset repair of $I$; for the lower bounds note that in our reductions $s\in \maxsub(X,\phi)$ if and only if $\maxsub(X,\phi)\neq \emptyset$.\footnote{In point of fact, the reduction of Lemma \ref{lem:eka} requires slight modification: remove assignments $(c_i,\top,\top)$ and add assignments $(c_i,j,k)$ for each assignment $(i,j,k)\in X$ where $i$ is an AND gate.}
 Considering consistent query answering, for the upper bounds we may apply the maximal subteam membership problem and for lower bounds we may simply use atomic queries.
That we may include also $\oplus$-repairs follows from the fact that each set of inclusion dependencies $\Sigma$ has a unique subset repair which is also the unique universal subset repair and the unique universal $\oplus$-repair \cite{CateFK15}. A database $U$ is a \emph{universal $*$-repair} of an inconsistent database $I$ if for each conjunctive query $q$, a tuple is a consistent answer to $q$ on $I$ if and only if it is an answer to $q$ on $U$ and contains only values that appear in $I$. That is, it only suffices to consult the universal repair for consistent answers.

\section{Complexity of Model Checking}
In this section we discuss the model checking problem for quantifier-free inclusion logic formulae. It turns out that the results of the previous section are now easily adaptable. 
 As above, we herein restrict attention to 
 quantifier-free formulae.

\begin{defi} Let $\phi\in \FOinc$. 
Then $\MC(\phi)$ is the problem of determining whether $\mA\models_X \phi$, given a model $\mA$ and a team $X$. 
\end{defi}

Hardness results for model checking can now be obtained by relating to maximal subteam.

\begin{lem}
\label{prop:up}
Let $\alpha,\beta\in \FOinc$ be such that 
\begin{enumerate}[(i)]
\item $\Fr{\alpha}\cap\Fr{\beta}=\emptyset$,
\item $\MR(\alpha)$ is $C$-hard for $C\in\{\LOGSPACE,\NL,\PTIME\}$, and
\item There is a team $Y$ of $\dom{\mA}$ with domain $\Fr{\beta}$ such that $\emptyset\neq \maxsub(\mA,Y,\beta)\subsetneq Y$.
\end{enumerate}
Then $\MC(\alpha\vee\beta)$ is $C$-hard.
\end{lem}
\begin{proof}
Let $(\mA,X,s)$ be an instance of $\MR(\alpha)$, that is, $\mA$ is a model, $X$ a team over $\Fr{\alpha}$ and $s\in X$. It suffices to 
define 
a first-order reduction from $(\mA,X,s)$ to a team $X'$ over $\Fr{\alpha}\cup \Fr{\beta}$ such that $s\in \maxsub(\mA,X,\alpha)$ iff $\mA\models_{X'} \alpha \vee\beta$. Let $Z_0:=\maxsub(\mA,Y,\beta)$ and $Z_1:=Y\setminus Z_0$. 
Note that by condition (i), the union of any $t\in X$ and $t'\in Y$ is an assignment over $\Fr{\alpha}\cup \Fr{\beta}$.
We define
\[X':=\{s\cup t'\mid t'\in Z_1\}\cup \{t\cup t'\mid t\in X\setminus \{s\},t'\in Z_0\}.\]
Since $Z_0$ and $Z_1$ are fixed, $X'$ is first-order 
definable
from $(\mA,X,s)$. By Locality (Proposition \ref{prop:locality}), 
we have $\maxsub(\mA,X',\alpha)\upharpoonright \Fr{\alpha} = \maxsub(\mA,X'\upharpoonright\Fr{\alpha},\alpha)=\maxsub(\mA,X,\alpha)$, and similarly $\maxsub(\mA,X',\beta)\upharpoonright \Fr{\beta} =\maxsub(\mA,Y,\beta)=Z_0$. 
Hence, it follows from Lemma~\ref{lem:maxsub} that $\mA\models_{X'} \alpha \vee\beta$ iff for all $t\cup t'\in X': t\in \maxsub(\mA,X,\alpha)\vee t'\in \maxsub(\mA,Y,\beta)$ iff $s\in \maxsub(\mA,X,\alpha)$. 
\end{proof}

Note that $\mA \models_X \phi$ if and only if $\maxsub(\mA,X,\phi)=X$ over inclusion logic formulae $\phi$. Hence, model checking can be reduced to maximal subteam membership tests over each individual assignment of a team. In particular, this means that model checking is at most as hard as maximal subteam membership; in some cases, as illustrated in Proposition \ref{prop:example}(\ref{eka}), it is strictly less hard.
\begin{lem}\label{lem:tomc}
Let $\alpha\in \FOinc$ be such that $\MR(\alpha)$ is in $C\in\{\LOGSPACE,\NL\}$. Then $\MC(\alpha)$ is in $C$.
\end{lem}

By Lemmata \ref{prop:dis}, \ref{prop:up}, \ref{lem:tomc}, Theorem \ref{thm:gfpcaptures}, and the results of the previous section, the computational complexity of model checking for various quantifier-free inclusion formulae directly follows. The following proposition illustrates some examples. 
Note that the semantics of the inclusion atom is clearly first-order expressible, and the same applies to any conjunction of inclusion atoms.
\begin{prop}\label{prop:example}
\text{ }\\\vspace{-5mm}
\begin{enumerate}[(a)]
\item $\MC(x\sub y)$ and $\MC(x\sub y\wedge u\sub v)$ are first-order definable.
\item $\MC(x\sub y \vee u\sub v)\text{ and }\MC(x\sub y \vee u=v)$ are $\NL$-complete.
\item $\MC((x\sub z\wedge y\sub z)\vee u\sub v)\text{ and }\MC((x\sub z\wedge y\sub z)\vee u=v)$ are $\PTIME$-complete.
\end{enumerate}
\end{prop}

\section{An $\NL$ Fragment of Inclusion Logic}\label{sect:nl}

Our aim in this section is to find a natural fragment of 
inclusion logic that
captures the complexity class $\NL$ over ordered finite models. 
Our approach is to consider preservation of $\NL$-computability under
the standard logical operators of $\FOinc$. By Lemma \ref{prop:dis}, 
we already know that $\NL$-computability of maximal 
subteam membership
is preserved under disjunctions. However, Theorem 
\ref{thm:trikotomia} shows 
that conjunction can increase the complexity of the maximal 
subteam membership problem from $\NL$ to $\PTIME$-complete, and by
Proposition \ref{prop:example}, combining a conjunction with a disjunction leads to
$\PTIME$-complete model-checking problems. Thus conjunction cannot be
used freely in the fragment we aim for. 

The following proposition shows
that a single universal quantifier can also increase complexity from 
$\NL$ to $\PTIME$-complete. In the proof we show $\PTIME$-hardness
by reduction from the $\PTIME$-complete problem GAME. An input to 
GAME is a DAG (directed acyclic graph) $G=(V,E)$ together with a node 
$a\in V$. Given such input $(V,E,a)$ we consider the following game 
$\game(V,E,a)$ between two players, I and II. 
During the game the players move a pebble along the edges of
$G$. In the initial position the pebble is on the node $a_0=a$. If after 
$2i$ moves the pebble is on a node $a_{2i}$, then player I chooses a node
$a_{2i+1}$ such that $(a_{2i},a_{2i+1})\in E$, and player II responds
by choosing a node $a_{2i+2}$ such that $(a_{2i+1},a_{2i+2})\in E$. 
The first player unable to move loses the game, and the other player 
wins it. Since $G$ is a DAG, every play of the game is finite.
In particular, the game is determined, i.e., one of the players 
has a winning strategy. Now we define
$(V,E,a)$ to be a positive instance of GAME if and only if player II has
a winning strategy in $\game(V,E,a)$. 

Note that GAME can be seen as a variation of the monotone circuit
value problem MCVP. Indeed, it is straighforward to define for
a given monotone circuit $C$ and input word $w$ an input $(V,E,a)$
for GAME such that $\game(V,E,a)$ simulates the evaluation game of
$C$ on $w$. Thus MCVP is logarithmic-space reducible to GAME. Conversely,
it is also easy to give a logarithmic-space reduction from GAME to
MCVP.


\begin{prop}
Let $\phi$ be the formula $\forall z(\lnot Eyz\lor z\subseteq x)$. Then
$\MR(\phi)$ is $\PTIME$-complete. Consequently, $\MC(\phi\lor Euv)$
is also $\PTIME$-complete.
\end{prop}

\begin{proof}
We give now a reduction from GAME to $\MR(\phi)$. Let $(V,E,a)$ be an 
input to GAME. Without loss of generality we assume that there is 
$b\in V$ such that $(b,a)\in E$. Now we simply let $\mA=(V,E)$, 
$X=\{s:\{x,y\}\to V\mid (s(x),s(y))\in E\}$ and $s_0=\{(x,b),(y,a)\}$.

We will use below the notation 
$I=\{c\in V\mid \forall d\in V:(c,d)\not\in E)\}$. Thus, $I$ consists 
of those elements $c\in V$ for which player II wins $\game(V,E,c)$ 
immediately because I cannot move. Furthermore, we denote by $W$
the set of all elements $c\in V$ such that player II has a winning
strategy in $\game(V,E,c)$.

Let $Y$ be the subteam of $X$ consisting of those assignments $s\in X$
for which $s(y)\in W$.
We will show that $Y=\maxsub(\mA,X,\phi)$. Hence in particular
$s_0\in\maxsub(\mA,X,\phi)$ if and only if $(V,E,a)$ is a positive 
instance of GAME, as desired.

To prove that $Y\subseteq\maxsub(\mA,X,\phi)$ it suffices to show that
$\mA\models_Y\phi$. Thus let $Z=Y[A/z]$, 
$Z'=\{s\in Z\mid (s(y),s(z))\not\in E\}$ and $Z''=(Z\setminus Z')
\cup Z_0$, where $Z_0=\{s\in Z\mid s(z)=s(x)\text{ and }s(y)\in I\}$. 
Then clearly $\mA\models_{Z'}\lnot Eyz$. To show that
$\mA\models_{Z''}z\subseteq x$ assume that $s\in Z''$. If 
$s\in Z\setminus Z'$, then $(s(y),s(z))\in E$, and since 
$s\restriction \{x,y\}\in Y$, player II has an answer $c$ to the move 
$s(z)$ of player I in $\game(V,E,s(y))$ such that $c\in W$. 
Thus, $s^*=\{(x,s(z)),(y,c)\}\in Y$. 
If $c\in I$, then $s^*(s^*(x)/z)\in Z_0$.
Otherwise there is some $d\in V$ such that $(c,d)\in E$, whence 
$s^*(d/z)\in Z\setminus Z'$. In both cases, there is $s'\in Z''$
such that $s'(x)=s(z)$. Assume then that $s\in Z_0$. Then by the definition
of $Z_0$ we have $s(x)=s(z)$. Thus we see that
for every $s\in Z''$ there is $s'\in Z''$ such that $s'(x)=s(z)$.
Now we can conclude that $\mA\models_Z \lnot Eyz\lor z\subseteq x$,
and hence $\mA\models_Y\phi$.

To prove that $\maxsub(\mA,X,\phi)\subseteq Y$ it suffices to show that
if $\mA\models_{Y'}\phi$ for a team $Y'\subseteq X$, then $s(y)\in W$ 
for every $s\in Y'$. Thus 
assume that $Y'$ satisfies $\phi$ and $s\in Y'$. We describe a winning
strategy for player II in $\game(V,E,s(y))$. If $s(y)\in I$
she has a trivial winning strategy. Otherwise player I is able to
move; let $c\in V$ be his first move. Since 
$\mA\models_{Y'}\phi$, there are $Z',Z''\subseteq Y'[A/z]$ such that
$Y'[A/z]=Z'\cup Z''$, $\mA\models_{Z'}\lnot Eyz$ and 
$\mA\models_{Z''}z\subseteq x$. Consider the assignment 
$s'=s(c/z)\in Y'[A/z]$. Since $(s'(y),s'(z))=(s(y),c)\in E$, it must 
be the case that $s'\in Z''$. Thus there is $s''\in Z''$ such that
$s''(x)=s'(z)=c$. Then the assignment $s^*=s''\restriction\{x,y\}$ is
in $Y'\subseteq X$, whence $(c,d)\in E$, where $d=s^*(y)$. Let $d$
be the answer of player II for the move $c$ of player I. We observe
now that using this strategy player II can find a legal answer 
from the set $\{s^*(y)\mid s^*\in Y'\}$ to any move of player I,
as long as player I is able to move. 
Since the game is determined, this is indeed a winning strategy.

Using Lemma \ref{prop:up}, we see that $\MC(\forall z(\lnot Eyz\lor 
z\subseteq x)\lor\beta)$ is $\PTIME$-hard for any non-trivial 
formula $\beta$ such that $x,y\not\in\fr(\beta)$, in particular for  
$\beta= Euv$.
\end{proof}
\begin{figure}
\centering
\begin{tabular}{ccc}
\adjustbox{valign=m}{
\begin{tikzpicture}[->,>=stealth,shorten >=1pt,auto,node distance=1.5cm,
 main node/.style={circle,draw,minimum size=0.6cm,draw,font=\sffamily\bfseries}]
\node[main node] (0) {$0$};
  \node[main node] (1) [right of =0]  {$a$};
 \node[main node] (2) [right of =1] {$1$};
  \node[main node] (5) [below of =1] {$2$};
    \node[main node] (6) [below of =2] {$3$};
   \node[main node] (3) [below of=5] {$b$};
   \node[main node] (4) [below of=6] {$4$};
\path[every node/.style={sloped,anchor=south,auto=false}]
(0) edge (1)
    (1) edge (2)
    (2) edge (5)
   (1) edge (5)
    (6) edge (5)
        (5) edge (3)
    (4) edge (3)
        (4) edge (5)
  (6) edge (4)
    (2) edge (6);
      
\end{tikzpicture}}
&
\hspace{.5cm}
$\mapsto$
&
\qquad
    \begin{tabular}{cccc}
    \cmidrule[.5pt]{2-4}
    &$x$ & $y$ & $z$ \\
    \cmidrule[.5pt]{2-4}
   $\bullet$&$0$ & $a$ & $1,2$ \\
    &$a$ & $1$ & $2,3$ \\
   $$ &$a$ & $2$ & $b$ \\
    $$&$1$ & $2$ & $b$ \\
    $\circ$&$1$ & $3$ & $2,4$ \\
    $\circ$&$2$ & $b$ & $2$ \\
    &$3$ & $2$ & $b$ \\
    &$3$ & $4$ & $2,b$ \\
    &$4$ & $2$ & $b$ \\
    $\circ$&$4$ & $b$ & $4$ \\
    \cmidrule[.5pt]{2-4}
    \end{tabular}
    \end{tabular}
\caption{GAME and $\MR(\forall z(\lnot Eyz\lor 
z\subseteq x))$\label{kuvaD}}
\end{figure}
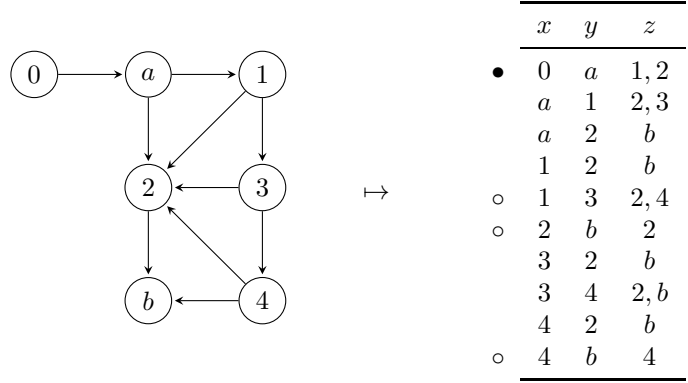

The example above shows that, similarly as conjunction, universal 
quantification cannot be freely used if the goal is to construct 
a fragment of inclusion logic that captures $\NL$. On the positive
side, we prove next that existential quantification preserves 
$\NL$-computability. Furthermore, we show that the same holds for
conjunction, provided that one of the conjuncts is in $\FO$.

\begin{lem}\label{exq-focon-maxsub}
Let  $\phi\in\FOinc$, $\psi\in \FO$, 
and let $X$ be a team of a model $\mA$. Then

(a) $\maxsub(\mA,X,\exists x\phi)=\{s\in X\mid 
s(a/x)\in X'\text{ for some }a\in A\}$,
where $X'=\maxsub(\mA,X[A/x],\phi)$, 

(b) $\maxsub(\mA,X,\phi\land\psi)=\maxsub(\mA,X',\phi)$, where
$X'=\maxsub(\mA,X,\psi)$.
\end{lem}

\begin{proof}
(a) Let $X'=\maxsub(\mA,X[A/x],\phi)$ and $X''=\{s\in X\mid 
s(a/x)\in X'\text{ for some }a\in A\}$. Assume that $Y\subseteq X$
is a team such that $\mA\models_Y\exists x\phi$. Then there is a function
$F: X  \rightarrow \Po(A)\setminus \{\emptyset\}$ such that
$\mA\models_{Y[F/x]}\phi$, and since clearly $Y[F/x]\subseteq X[A/x]$,
we have $Y[F/x]\subseteq X'$. Thus for every $s\in Y$ there is $a\in A$
such that $s(a/x)\in X'$, and hence we see that $Y\subseteq X''$.
In particular $\maxsub(\mA,X,\exists x\phi)\subseteq X''$. 
To prove the converse inclusion it suffices to show that 
$\mA\models_{X''}\exists x\phi$. Let $G: X''  \rightarrow 
\Po(A)\setminus \{\emptyset\}$ be the function defined by 
$G(s)=\{a\in A\mid s(a/x)\in X'\}$. By the definition of $X''$,
this function is well-defined. It is now easy to see that 
$X''[G/x]=X'$, whence $\mA\models_{X''[G/x]}\phi$, as desired.

(b) Let $X'=\maxsub(\mA,X,\psi)$ and $X''=\maxsub(\mA,X',\psi)$. 
Assume first that $Y\subseteq X$ is a team such that 
$\mA\models_Y\phi\land\psi$. Then $\mA\models_Y\psi$, 
whence $Y\subseteq X'$, and furthermore $Y\subseteq X''$, 
since $\mA\models_Y\phi$. In particular, 
$\maxsub(\mA,X,\phi\land\psi)\subseteq X''$. 
On the other hand, by definition $\mA\models_{X''}\phi$. Similarly
$\mA\models_{X'}\psi$, whence by downward closure of $\FO$ (Corollary \ref{cor:dc}), $\mA\models_{X''}\psi$. Thus we see that
$\mA\models_{X''}\phi\land\psi$, which implies that 
$X''\subseteq\maxsub(\mA,X,\phi\land\psi)$.
\end{proof}

As a straightforward corollary to this lemma we obtain the following complexity preservation result.

\begin{prop}\label{NL-pres}
Let  $\phi\in\FOinc$, $\psi\in \FO$, and assume that $\MR(\phi)$ is in
a complexity class $C\in\{\LOGSPACE,\NL\}$. Then

(a) $\MR(\exists x\phi)$ is in $C$, and

(b) $\MR(\phi\land\psi)$ is in $C$.
\end{prop}

\begin{proof}
(a) By Lemma \ref{exq-focon-maxsub}(a), to check whether a given assignment
$s$ is in $\maxsub(\mA,X,\exists x\phi)$ it suffices to check
whether $s(a/x)$ is in $\maxsub(\mA,X[A/x],\phi)$ for some $a\in A$.
Clearly this task can be done in $C$ assuming that $\MR(\phi)$ is in $C$.

(b) By Lemma \ref{exq-focon-maxsub}(b), it suffices to show that
the problem whether an assignment $s$ is in $\maxsub(\mA,X',\phi)$,
where $X'=\maxsub(\mA,X,\psi)$, can be solved in $C$ with respect to
the input $(s,\mA,X)$. Since $\psi\in\FO$, the team $X'$ can be computed
in $C$, whence the claim follows from the assumption that 
$\MR(\phi)$ is in $C$.
\end{proof}

Summarising Lemma \ref{prop:dis} and Proposition \ref{NL-pres}, 
$\NL$-computability
of maximal subteam membership is preserved by disjunction, conjunction
with first-order formulas, and existential quantification. Since 
maximal subteam problem is in $\NL$ for all first-order formulas and,
by Lemma \ref{thm:NLarb}, for all inclusion atoms, we define our fragment
$\NLFOinc$ of inclusion logic by the following grammar:
\[\phi::=\alpha\mid \tuple x\sub \tuple y \mid \phi\vee\phi \mid\phi\wedge\alpha\mid \exists x\phi,\]
where $\alpha\in \FO$.

\begin{thm}
$\MC(\phi)$ is in $\NL$ for every $\phi\in\NLFOinc$. 
\end{thm}

\begin{proof}
By an easy induction we see that $\MR(\phi)$ is in $\NL$ for every 
$\phi\in\NLFOinc$. The claim follows now from Lemma \ref{lem:tomc}.
\end{proof}

Vice versa, to show that each $\NL$ property of ordered models can be expressed in $\NLFOinc$, it suffices to show that $\TC$ translates to $\NLFOinc$ over ordered models.
\begin{thm}
\label{thm:fromTC}
Over finite ordered models, $\TC\leq \NLFOinc$.
\end{thm}

\begin{proof}
 By Theorem \ref{thm:TCcaptures} we may assume without loss of generality that any $\TC$ sentence $\phi$ is of the form 
$[\TC_{\tuple x,\tuple y}\alpha(\tuple x,\tuple y)](\tuple \min,\tuple \max)$ 
where $\tuple x$ and $\tuple y$ are $n$-ary sequences of variables. We define an equivalent $\NLFOinc$ sentence $\phi'$ as follows:
\begin{equation}\label{lause}
\phi':=\exists \tuple x\tuple y \tuple t_x\tuple t_y (\psi_1\wedge \psi_2\wedge \psi_3\wedge\psi_4)
\end{equation}
where 
\begin{itemize}
\item $\psi_1:=\,\tuple y\tuple t_y\sub \tuple x\tuple t_x$,
\item $ \psi_2:=\,(\tuple t_x <\tuple \max \wedge \tuple t_x < \tuple t_y\wedge \alpha(\tuple x,\tuple y)) \vee (\tuple t_x=\tuple \max \wedge \tuple t_y =\tuple \min)$,
\item $\psi_3:=\,\neg \tuple t_x=\tuple \min\vee \tuple x=\tuple \min $, and
\item $\psi_4 :=\,\neg \tuple t_x=\tuple \max\vee \tuple x= \tuple \max$.
\end{itemize}
For two tuples of variables $\tuple x$ and $\tuple y$ of the same length, we write $\tuple x < \tuple y$ as a shorthand for the formula expressing that $\tuple x$ is less than $\tuple y$ in the induced lexicographic ordering, and $\tuple x = \tuple y$ for the conjunction expressing that $\tuple x$ and $\tuple y$ are pointwise identical. Observe that in \eqref{lause} the 
tuple
$\tuple t_x$ can be thought of as a counter which indicates an upper bound for the $\alpha$-path distance of $\tuple x$ from $\tuple \min$.

Assuming
$\mA\models \phi'$, 
we find a non-empty team $X$ such that $\mA \models_X  \psi_1\wedge \psi_2\wedge \psi_3\wedge\psi_4$. Now, $\mA \models_X  \psi_1\wedge \psi_2$ entails that there is an assignment $s\in X$ mapping $\tuple t_x$ to $\tuple \min$, and $\mA \models_X  \psi_3$ implies that $s$ maps $\tuple x$ to $\tuple \min$, too. Then $\mA \models_X  \psi_1\wedge \psi_2$ entails that there is an $\alpha$-path from $\tuple \min$ to $s'(\tuple x)$ for some $s'\in X$ with $s'(\tuple t_x)=\tuple \max$. Lastly, by $\mA \models_X  \psi_4$ it follows that $s'(\tuple x) = \tuple \max$ which shows that 
$[\TC_{\tuple x,\tuple y}\alpha(\tuple x,\tuple y)](\tuple \min,\tuple \max)$.

Assume then that  
$[\TC_{\tuple x,\tuple y}\alpha(\tuple x,\tuple y)](\tuple \min,\tuple \max)$, 
that is, there is an $\alpha$-path $\tuple v_1, \ldots ,\tuple v_k$ where $\tuple v_1=\tuple \min$ and $\tuple v_k=\tuple \max$. We may assume that there are no cycles in the path. Let $\tuple a_i$ denote the $i$th element in the lexicographic ordering of $A^n$. 
 We let $X=\{s_1, \ldots ,s_k\}$ be such that $(\tuple x,\tuple y ,\tuple t_x,\tuple t_y)$ is mapped   to $( \tuple v_i, \tuple v_{i+1} ,\tuple a_i ,\tuple a_{i+1})$ by $s_i$, for $i=1, \ldots ,k-2$,  to $( \tuple v_{k-1}, \tuple v_{k} ,\tuple a_{k-1} ,\tuple \max)$ by $s_{k-1}$, and  to $( \tuple v_k, \tuple v_1 ,\tuple \max ,\tuple \min)$ by $s_{k}$. It is straightforward to verify that $\mA\models_X \psi_1\wedge \psi_2\wedge \psi_3\wedge\psi_4$ from which it follows that $\mA\models \phi'$. 
\end{proof}

It now follows by the above two theorems and Theorem \ref{thm:TCcaptures} that $ \NLFOinc$ captures $\NL$.
\begin{thm}
A class $\calC$ of finite ordered models is in $\NL$ iff it can be defined in $\NLFOinc$. 
\end{thm}
\section{Conclusion}

We have studied the complexity of inclusion logic from the vantage point of two computational problems: the maximal subteam membership and the model checking problems for fixed inclusion logic formulae. We gave a complete characterization for the former in terms of arbitrary conjunctions/disjunctions of unary inclusion atoms. In particular, we showed that maximal subteam membership is $\PTIME$-complete for any conjunction of unary inclusion atoms, provided that the conjunction contains two non-trivial atoms that are not inverses of each other. Using these results we characterized the complexity of model checking for several 
quantifier-free inclusion logic formulae. 
We also presented a safety game for the maximal subteam problem and used it to demonstrate that the problem is less complex if the range of inputs is restricted to teams on which the inclusion atoms reference a key.  We leave it for future research to address the complexity of 
quantifier-free inclusion logic formulae that involve inclusion atoms of higher arity and both disjunctions and conjunctions. 

Assuming the presence of quantifiers we presented a simple universally quantified formula that has $\PTIME$-complete maximal subteam membership problem. Finally, we defined a fragment of inclusion logic, obtained by restricting the scope of conjunction and universal quantification, that captures non-deterministic logarithmic space over finite ordered models. 

\bibliography{biblio}

\begin{thebibliography}{10}

\bibitem{AfratiK09}
Foto~N. Afrati and Phokion~G. Kolaitis.
\newblock Repair checking in inconsistent databases: algorithms and complexity.
\newblock In {\em Database Theory - {ICDT} 2009, 12th International Conference,
  St. Petersburg, Russia, March 23-25, 2009, Proceedings}, pages 31--41, 2009.
\newblock URL: \url{http://doi.acm.org/10.1145/1514894.1514899}, \href
  {http://dx.doi.org/10.1145/1514894.1514899}
  {\path{doi:10.1145/1514894.1514899}}.

\bibitem{ArenasBC99}
Marcelo Arenas, Leopoldo~E. Bertossi, and Jan Chomicki.
\newblock Consistent query answers in inconsistent databases.
\newblock In {\em Proceedings of the Eighteenth {ACM} {SIGACT-SIGMOD-SIGART}
  Symposium on Principles of Database Systems, May 31 - June 2, 1999,
  Philadelphia, Pennsylvania, {USA}}, pages 68--79, 1999.
\newblock URL: \url{http://doi.acm.org/10.1145/303976.303983}, \href
  {http://dx.doi.org/10.1145/303976.303983} {\path{doi:10.1145/303976.303983}}.

\bibitem{ChomickiM05}
Jan Chomicki and Jerzy Marcinkowski.
\newblock Minimal-change integrity maintenance using tuple deletions.
\newblock {\em Inf. Comput.}, 197(1-2):90--121, 2005.
\newblock URL: \url{https://doi.org/10.1016/j.ic.2004.04.007}, \href
  {http://dx.doi.org/10.1016/j.ic.2004.04.007}
  {\path{doi:10.1016/j.ic.2004.04.007}}.

\bibitem{CoranderHKPV16}
Jukka Corander, Antti Hyttinen, Juha Kontinen, Johan Pensar, and Jouko
  V{\"{a}}{\"{a}}n{\"{a}}nen.
\newblock A logical approach to context-specific independence.
\newblock In {\em Logic, Language, Information, and Computation - 23rd
  International Workshop, WoLLIC 2016, Puebla, Mexico, August 16-19th, 2016.
  Proceedings}, pages 165--182, 2016.
\newblock URL: \url{https://doi.org/10.1007/978-3-662-52921-8_11}, \href
  {http://dx.doi.org/10.1007/978-3-662-52921-8_11}
  {\path{doi:10.1007/978-3-662-52921-8_11}}.

\bibitem{Durand2018}
Arnaud Durand, Miika Hannula, Juha Kontinen, Arne Meier, and Jonni Virtema.
\newblock Approximation and dependence via multiteam semantics.
\newblock {\em Annals of Mathematics and Artificial Intelligence}, Jan 2018.
\newblock URL: \url{https://doi.org/10.1007/s10472-017-9568-4}, \href
  {http://dx.doi.org/10.1007/s10472-017-9568-4}
  {\path{doi:10.1007/s10472-017-9568-4}}.

\bibitem{DurandKRV15}
Arnaud Durand, Juha Kontinen, Nicolas de~Rugy{-}Altherre, and Jouko
  V{\"{a}}{\"{a}}n{\"{a}}nen.
\newblock Tractability frontier of data complexity in team semantics.
\newblock In {\em Proceedings Sixth International Symposium on Games, Automata,
  Logics and Formal Verification, GandALF 2015, Genoa, Italy, 21-22nd September
  2015.}, pages 73--85, 2015.
\newblock URL: \url{https://doi.org/10.4204/EPTCS.193.6}, \href
  {http://dx.doi.org/10.4204/EPTCS.193.6} {\path{doi:10.4204/EPTCS.193.6}}.

\bibitem{ebbing13}
Johannes Ebbing, Lauri Hella, Arne Meier, Julian{-}Steffen M{\"{u}}ller, Jonni
  Virtema, and Heribert Vollmer.
\newblock Extended modal dependence logic.
\newblock In {\em Logic, Language, Information, and Computation - 20th
  International Workshop, WoLLIC 2013, Darmstadt, Germany, August 20-23, 2013.
  Proceedings}, pages 126--137, 2013.
\newblock URL: \url{http://dx.doi.org/10.1007/978-3-642-39992-3_13}, \href
  {http://dx.doi.org/10.1007/978-3-642-39992-3_13}
  {\path{doi:10.1007/978-3-642-39992-3_13}}.

\bibitem{ebbing14}
Johannes Ebbing, Juha Kontinen, Julian{-}Steffen M{\"{u}}ller, and Heribert
  Vollmer.
\newblock A fragment of dependence logic capturing polynomial time.
\newblock {\em Logical Methods in Computer Science}, 10(3), 2014.
\newblock URL: \url{http://dx.doi.org/10.2168/LMCS-10(3:3)2014}, \href
  {http://dx.doi.org/10.2168/LMCS-10(3:3)2014}
  {\path{doi:10.2168/LMCS-10(3:3)2014}}.

\bibitem{galliani12}
Pietro Galliani.
\newblock Inclusion and exclusion dependencies in team semantics: On some
  logics of imperfect information.
\newblock {\em Annals of Pure and Applied Logic}, 163(1):68 -- 84, 2012.

\bibitem{gallhella13}
Pietro Galliani and Lauri Hella.
\newblock {Inclusion Logic and Fixed Point Logic}.
\newblock In Simona Ronchi~Della Rocca, editor, {\em Computer Science Logic
  2013 (CSL 2013)}, volume~23 of {\em Leibniz International Proceedings in
  Informatics (LIPIcs)}, pages 281--295, Dagstuhl, Germany, 2013. Schloss
  Dagstuhl--Leibniz-Zentrum fuer Informatik.
\newblock URL: \url{http://drops.dagstuhl.de/opus/volltexte/2013/4203}, \href
  {http://dx.doi.org/http://dx.doi.org/10.4230/LIPIcs.CSL.2013.281}
  {\path{doi:http://dx.doi.org/10.4230/LIPIcs.CSL.2013.281}}.

\bibitem{gradel12}
Erich Gr\"adel.
\newblock Model-checking games for logics of imperfect information.
\newblock {\em Theoretical Computer Science (to appear)}, 2012.
\newblock URL:
  \url{http://www.sciencedirect.com/science/article/pii/S0304397512009541},
  \href {http://dx.doi.org/10.1016/j.tcs.2012.10.033}
  {\path{doi:10.1016/j.tcs.2012.10.033}}.

\bibitem{Gradel16}
Erich Gr{\"{a}}del.
\newblock Games for inclusion logic and fixed-point logic.
\newblock In {\em Dependence Logic, Theory and Applications}, pages 73--98.
  2016.
\newblock URL: \url{https://doi.org/10.1007/978-3-319-31803-5_5}, \href
  {http://dx.doi.org/10.1007/978-3-319-31803-5_5}
  {\path{doi:10.1007/978-3-319-31803-5_5}}.

\bibitem{HannulaK16}
Miika Hannula and Juha Kontinen.
\newblock A finite axiomatization of conditional independence and inclusion
  dependencies.
\newblock {\em Inf. Comput.}, 249:121--137, 2016.
\newblock URL: \url{https://doi.org/10.1016/j.ic.2016.04.001}, \href
  {http://dx.doi.org/10.1016/j.ic.2016.04.001}
  {\path{doi:10.1016/j.ic.2016.04.001}}.

\bibitem{HannulaKV18}
Miika Hannula, Juha Kontinen, and Jonni Virtema.
\newblock Polyteam semantics.
\newblock In {\em Logical Foundations of Computer Science - International
  Symposium, {LFCS} 2018, Deerfield Beach, FL, USA, January 8-11, 2018,
  Proceedings}, pages 190--210, 2018.
\newblock URL: \url{https://doi.org/10.1007/978-3-319-72056-2_12}, \href
  {http://dx.doi.org/10.1007/978-3-319-72056-2_12}
  {\path{doi:10.1007/978-3-319-72056-2_12}}.

\bibitem{HannulaKVV18}
Miika Hannula, Juha Kontinen, Jonni Virtema, and Heribert Vollmer.
\newblock Complexity of propositional logics in team semantic.
\newblock {\em {ACM} Trans. Comput. Log.}, 19(1):2:1--2:14, 2018.
\newblock URL: \url{http://doi.acm.org/10.1145/3157054}, \href
  {http://dx.doi.org/10.1145/3157054} {\path{doi:10.1145/3157054}}.

\bibitem{hodges97}
Wilfrid Hodges.
\newblock {C}ompositional {S}emantics for a {L}anguage of {I}mperfect
  {I}nformation.
\newblock {\em Journal of the Interest Group in Pure and Applied Logics}, 5
  (4):539--563, 1997.

\bibitem{immerman86}
Neil Immerman.
\newblock Relational queries computable in polynomial time.
\newblock {\em Information and control}, 68(1):86--104, 1986.

\bibitem{Immerman83}
Neil Immerman.
\newblock Languages that capture complexity classes.
\newblock {\em {SIAM} J. Comput.}, 16(4):760--778, 1987.

\bibitem{Immerman88}
Neil Immerman.
\newblock Nondeterministic space is closed under complementation.
\newblock {\em {SIAM} J. Comput.}, 17(5):935--938, 1988.
\newblock URL: \url{https://doi.org/10.1137/0217058}, \href
  {http://dx.doi.org/10.1137/0217058} {\path{doi:10.1137/0217058}}.

\bibitem{kontinenj13}
Jarmo Kontinen.
\newblock Coherence and computational complexity of quantifier-free dependence
  logic formulas.
\newblock {\em Studia Logica}, 101(2):267--291, 2013.
\newblock URL: \url{http://dx.doi.org/10.1007/s11225-013-9481-8}, \href
  {http://dx.doi.org/10.1007/s11225-013-9481-8}
  {\path{doi:10.1007/s11225-013-9481-8}}.

\bibitem{KontinenKLV14}
Juha Kontinen, Antti Kuusisto, Peter Lohmann, and Jonni Virtema.
\newblock Complexity of two-variable dependence logic and if-logic.
\newblock {\em Inf. Comput.}, 239:237--253, 2014.
\newblock URL: \url{https://doi.org/10.1016/j.ic.2014.08.004}, \href
  {http://dx.doi.org/10.1016/j.ic.2014.08.004}
  {\path{doi:10.1016/j.ic.2014.08.004}}.

\bibitem{KoutrisW17}
Paraschos Koutris and Jef Wijsen.
\newblock Consistent query answering for self-join-free conjunctive queries
  under primary key constraints.
\newblock {\em {ACM} Trans. Database Syst.}, 42(2):9:1--9:45, 2017.
\newblock URL: \url{http://doi.acm.org/10.1145/3068334}, \href
  {http://dx.doi.org/10.1145/3068334} {\path{doi:10.1145/3068334}}.

\bibitem{Pacuit2016}
Eric Pacuit and Fan Yang.
\newblock {\em Dependence and Independence in Social Choice: Arrow's Theorem},
  pages 235--260.
\newblock Springer International Publishing, Cham, 2016.
\newblock URL: \url{https://doi.org/10.1007/978-3-319-31803-5_11}, \href
  {http://dx.doi.org/10.1007/978-3-319-31803-5_11}
  {\path{doi:10.1007/978-3-319-31803-5_11}}.

\bibitem{CateFK15}
Balder ten Cate, Ga{\"{e}}lle Fontaine, and Phokion~G. Kolaitis.
\newblock On the data complexity of consistent query answering.
\newblock {\em Theory Comput. Syst.}, 57(4):843--891, 2015.
\newblock URL: \url{https://doi.org/10.1007/s00224-014-9586-0}, \href
  {http://dx.doi.org/10.1007/s00224-014-9586-0}
  {\path{doi:10.1007/s00224-014-9586-0}}.

\bibitem{vaananen07}
Jouko V\"a\"an\"anen.
\newblock {\em Dependence Logic}.
\newblock Cambridge University Press, 2007.

\bibitem{vaananen08b}
Jouko V\"a\"an\"anen.
\newblock {M}odal {D}ependence {L}ogic.
\newblock In Krzysztof~R. Apt and Robert van Rooij, editors, {\em New
  Perspectives on Games and Interaction}. Amsterdam University Press,
  Amsterdam, 2008.

\bibitem{vardi82}
Moshe~Y Vardi.
\newblock The complexity of relational query languages.
\newblock In {\em Proceedings of the fourteenth annual ACM symposium on Theory
  of computing}, pages 137--146. ACM, 1982.

\bibitem{Virtema17}
Jonni Virtema.
\newblock Complexity of validity for propositional dependence logics.
\newblock {\em Inf. Comput.}, 253:224--236, 2017.
\newblock URL: \url{https://doi.org/10.1016/j.ic.2016.07.008}, \href
  {http://dx.doi.org/10.1016/j.ic.2016.07.008}
  {\path{doi:10.1016/j.ic.2016.07.008}}.

\bibitem{Vollmerbook}
Heribert Vollmer.
\newblock {\em Introduction to Circuit Complexity - {A} Uniform Approach}.
\newblock Texts in Theoretical Computer Science. An {EATCS} Series. Springer,
  1999.
\newblock URL: \url{https://doi.org/10.1007/978-3-662-03927-4}, \href
  {http://dx.doi.org/10.1007/978-3-662-03927-4}
  {\path{doi:10.1007/978-3-662-03927-4}}.

\bibitem{YangV16}
Fan Yang and Jouko V{\"{a}}{\"{a}}n{\"{a}}nen.
\newblock Propositional logics of dependence.
\newblock {\em Ann. Pure Appl. Logic}, 167(7):557--589, 2016.
\newblock URL: \url{https://doi.org/10.1016/j.apal.2016.03.003}, \href
  {http://dx.doi.org/10.1016/j.apal.2016.03.003}
  {\path{doi:10.1016/j.apal.2016.03.003}}.

\end{thebibliography}

\end{document}